\documentclass{amsart}
\usepackage{amssymb,amscd,amsthm}
\usepackage{latexsym}
\usepackage[final]{graphicx}

\newcommand{\Z}{{\mathbb Z}}
\newcommand{\R}{{\mathbb R}}
\newcommand{\C}{{\mathbb C}}
\newcommand{\N}{{\mathbb N}}

\newcommand{\E}{{\mathbb E}}
\newcommand{\PP}{{\mathbb P}}

\newtheorem{theorem}{Theorem}
\newtheorem{remark}{Remark}[section]
\newtheorem{lemma}[remark]{Lemma}
\newtheorem{proposition}[remark]{Proposition}
\newtheorem{coro}[remark]{Corollary}

\begin{document}

\title[Anderson Localization]{An Introduction to the\\ Mathematics of Anderson Localization }

\author{G\"unter Stolz}

\address{Department of Mathematics, University of Alabama at Birmingham, Birmingham, AL~35294, USA}

\email{stolz@math.uab.edu}

\thanks{This work was supported in part by NSF grant DMS-0653374.}

\begin{abstract}
We give a widely self-contained introduction to the mathematical
theory of the Anderson model. After defining the Anderson model and
determining its almost sure spectrum, we prove localization
properties of the model. Here we discuss spectral as well as
dynamical localization and provide proofs based on the fractional
moments (or Aizenman-Molchanov) method.

We also discuss, in less self-contained form, the extension of the
fractional moment method to the continuum Anderson model. Finally,
we mention major open problems.

These notes are based on several lecture series which the author
gave at the Kochi School on Random Schr\"odinger Operators, November
26-28, 2009, the Arizona School of Analysis and Applications, March
15-19, 2010 and the Summer School on Mathematical Physics, Sogang
University, July 20-23, 2010.
\end{abstract}

\maketitle

\section{Introduction}

In 1958 the physicist P.\ W.\ Anderson introduced the model which is
now named after him  to explain the quantum mechanical effects of
disorder, as present in materials such as alloys and amorphous media
\cite{Anderson}. The most famous phenomena which arise in the
context of this model are {\it Anderson localization}, i.e.\ the
suppression of electron transport due to disorder, and the {\it
Anderson transition} in three-dimensional disordered media which
predicts the existence of a mobility edge separating energy regions
of localized states from an extended states region. Anderson
localization has important consequences throughout physics, in
theory and experiment. Anderson's work, and that of N.\ F.\ Mott and
J.\ H.\ van Vleck, won the 1977 physics Nobel prize ``for their
fundamental theoretical investigations of the electronic structure
of magnetic and disordered
systems"\footnote{http://nobelprize.org/nobel\_prizes/physics/laureates/1977/}.

Mathematically rigorous studies of the Anderson Model and other
models of random operators started in the 1970s, with the first
proof of Anderson localization for a related one-dimensional model
provided by I.\ Goldsheid, S.\ Molchanov and L.\ Pastur in 1977 \cite{GMP},
followed several years later by a proof of localization for the
actual Anderson model by H.\ Kunz and B.\ Souillard \cite{KS}, also
initially for dimension one. Since then the study of random
operators has become an important field of mathematical physics,
which has led to a tremendous amount of research activity and many
mathematical results.

While the Anderson transition and extended states are still an open
mathematical challenge, by now a good rigorous understanding of
Anderson localization has been achieved. Several powerful methods
have been found to prove Anderson localization. Important
differences exist between one-dimensional and multi-dimensional
models, where different physical mechanisms are responsible for
localization effects. In these notes we will focus on methods which
allow to prove Anderson localization in arbitrary dimension. Two
such methods are available: The method of {\it multiscale analaysis}
(MSA) developed in 1983 by Fr\"ohlich and Spencer
\cite{Frohlich/Spencer}, and the {\it fractional moments method}
(FMM) introduced by Aizenman and Molchanov in 1993
\cite{Aizenman/Molchanov}.

MSA has produced results in situations which are out of reach for an
approach through the FMM, see Section~\ref{sec:singular} for some related discussion. However, the FMM is mathematically more elementary, in
particular for the case of the classical {\it discrete} Anderson
model which will be our main focus here. Also, under suitable
assumptions, the FMM allows to prove stronger results on dynamical
localization than can be obtained by MSA. Therefore, in these
lectures, after an introduction to the Anderson model and its basic
spectral properties, we will discuss how to prove Anderson
localization based on the FMM.

After more than 50 years of physical research and more than 30 years
of mathematical work a vast literature with results on Anderson
localization and, more generally, the physics of disordered quantum
mechanical systems, is available. In these introductory lectures we
ignore most of the literature as it can not be our goal to provide a
comprehensive survey, not even of the mathematical research which
has been done. Some book length presentations, or parts of such,
which provide very good further reading and many more references are
\cite{CFKS, Carmona/Lacroix, Pastur/Figotin, Stollmann, Kirsch}.

It is assumed below that the reader is familiar with measure and
integration theory as presented in \cite{Rudin}, with basic probabilistic concepts such as independence, and with the
foundations of the theory of linear operators in Hilbert spaces, up
to the spectral theorem for self-adjoint operators and consequences
such as spectral types (absolutely continuous, singular continuous
and pure point spectrum) and the abstract solution of the
time-dependent Schr\"odinger equation via Stone's theorem, e.g.\
\cite{Weidmann} or \cite{Reed/Simon}. Otherwise, we have tried to
keep these notes mostly self-contained. For much of the first seven sections we provide full proofs.

We do not aim at the most general known results, but rather want to
demonstrate that simple and natural mathematical ideas can be used
to rigorously establish Anderson localization. Many further
developments of the ideas discussed here can be found in the
literature. The references provided below can serve as a starting
point for further reading. An ideal source for continued reading and
learning the state of the art of much what is discussed here will be
the upcoming book \cite{AizenmanWarzelBook} by M.\ Aizenman and S.\
Warzel.

In Section~\ref{sec:model} we introduce the Anderson model and, as a warm-up, prove its first important property, namely that its spectrum is almost surely deterministic. The rest of these notes exclusively deals with the phenomenon of Anderson localization. Section~\ref{sec:locproperties} introduces the concepts of {\it spectral localization} and {\it dynamical localization}, followed by a discussion of what is known on the physics level of rigor.

In Sections~\ref{sec:largedisorder} and \ref{sec:Grafmethod} we prove localization in the large disorder regime of the Anderson model. This is done via the fractional moments method, by first proving in Section~\ref{sec:largedisorder} that fractional moments of Green's function decay exponentially, and by then showing in Section~\ref{sec:Grafmethod} that this implies dynamical as well as spectral localization. In these sections we use methods which were developed in some of the first papers on the fractional moments method, e.g.\ \cite{Aizenman/Molchanov} and \cite{Graf94}. In particular, these methods work directly for the Anderson model in {\it infinite volume}.

Subsequently, other methods were introduced in the literature, e.g.\ \cite{Aizenman94} or \cite{ASFH01}, which use {\it finite volume} restrictions of the Anderson model. A central concept here are so-called {\it eigenfunction correlators}. These methods have proven to be very powerful in further-reaching work, for example in dealing with the continuum Anderson model or multi-particle Anderson models. Thus we introduce this approach in Section~\ref{sec:efcor} and Appendix~\ref{sec:appendix} and show how they yield an alternative proof of localization.

Section~\ref{sec:bandedge} discusses the second main regime in which multi-dimensional localization has been established rigorously, the band edge regime. Among the new ideas needed here are the phenomenon of {\it Lifshits tails} of the integrated density of states near spectral edges and a geometric decoupling method to control correlations in Green's function. Parts of this section have the character of an outline, referring to the literature for some of the results used.

Entirely written in form of an outline is Section~\ref{sec:continuum}, in which we
discuss the extension of the FMM to continuum Anderson models, as
accomplished in \cite{AENSS06} and \cite{BNSS06}. This requires
considerable technical effort and we merely point out the difficulties
which had to be overcome and mention some of the tools which allowed to accomplish this.

The Anderson model and, more generally, the quantum mechanics of
disordered media, provides many difficult future challenges for
mathematicians. We discuss some of them in our concluding
Section~\ref{sec:problems}.

\vspace{.3cm}

\noindent {\bf Acknowledgement:} The author's knowledge of random
operators and, in particular, of the Anderson model has benefitted
from many other mathematicians, through their works as well as
through personal contact. We apologize for not being able to
properly give credit for the origins of all of these benefits. But
we need and want to make an exception for the contributions of
Michael Aizenman, who was the driving force in the development of
the fractional moments method and has influenced the author's way of
thinking about random operators in multiple ways. Much of what we
have to say here is based on ideas of Michael and his collaborators.
In particular, special thanks are due to Michael and Simone Warzel
for letting the author use some preliminary material from
\cite{AizenmanWarzelBook} in Section~\ref{sec:efcor} and
Appendix~\ref{sec:appendix} below.

Thanks are also due to a referee for useful suggestions which improved our presentation.

Finally, the author would like to thank the organizers of the Kochi School on Random Schr\"odinger Operators, the Arizona School of Analysis and Applications and the Summer School on Mathematical Physics at Sogang
University for their invitations. Without these opportunities to lecture on the material covered here these notes would never have been written.

\section{The Anderson Model} \label{sec:model}

\subsection{The Discrete Laplacian}

Below we will introduce the Anderson model as a {\it discrete}
Schr\"odinger operator, acting, for dimension $d\ge 1$, on the
Hilbert space
\[ \ell^2(\Z^d) = \{ u:\Z^d \to \C: \,\sum_{n\in \Z^d} |u(n)|^2 < \infty \}, \]
with inner product $\langle u, v \rangle = \sum_n \overline{u(n)}
v(n)$.

The usual negative Laplacian $-\Delta = -\sum_j
\partial^2/\partial^2 x_j$ is replaced by its discrete analogue
$h_0$, which acts on $u\in \ell^2(\Z^d)$ by
\begin{equation} \label{eq:Laplace}
 (h_0 u)(n) = - \sum_{k\in \Z^d,\,|k|=1} u(n+k), \quad n\in \Z^d,
 \end{equation}
where $|k|=|k_1|+\ldots+|k_d|$ is the graph distance on $\Z^d$. More
appropriately, a finite difference approximation of $-\Delta$ would
be given by $h_0+2d$, but we neglect the mathematically trivial
shift by $2d$ (which should still be kept in mind for physical
interpretations). In physics, the Hamiltonian $h_0$ (or its
negative) most frequently arises more directly, i.e.\ not as a discretization
of a differential operator, in the context of the so-called tight-binding approximation.

Just as the continuum Laplacian, the discrete Laplacian is unitarily
equivalent to a multiplication operator via Fourier transform. Here
we consider the Fourier transform
\[ F: L^2([0,2\pi)^d) \to \ell^2(\Z^d),\]
which is the unitary operator given by
\[ (Fg)(n) = (2\pi)^{-d/2} \int_{[0,2\pi)^d} g(x) e^{-ix\cdot n} \,dx,\]
with inverse
\[ (F^{-1} u)(x) = \mbox{l.i.m.} \:(2\pi)^{-d/2} \sum_{n\in \Z^d, |n|\le N} u(n) e^{ix\cdot n}.\]
Here $x\cdot n = x_1 n_1 + \ldots x_d n_d$ and l.i.m.\ denotes the
limit $N\to\infty$ in $\ell^2(\Z^d)$.

A calculation shows that
\[ F^{-1} h_0 F = - 2 \sum_{j=1}^d \cos(x_j),\]
where the right-hand side is understood as a multiplication operator
on $L^2([0,2\pi)^d)$ in the variable $x=(x_1,\ldots,x_d)$. The
function $g(x)=-2\sum_j \cos(x_j)$ is real-valued and bounded. Thus
$h_0$ is bounded and self-adjoint (which also can be checked
directly from the definition (\ref{eq:Laplace}) without use of the
Fourier transform). The range of $g$ gives the spectrum of $h_0$,
\begin{equation} \label{eq:Laplacespec}
 \sigma(h_0) = [-2d,2d].
 \end{equation}
With a bit more effort one can show that inverse images of
Lebesgue-nullsets in $\R$ under the function $g$ are
Lebesgue-nullsets in $\R^d$. Thus the spectrum of $h_0$ is purely
absolutely continuous.

Another similarity of $h_0$ with the continuum Laplacian is that is
has plane waves as generalized eigenfunctions. To see this, let
$x\in [0,2\pi)^d$ and set
\begin{equation} \label{eq:genef}
\phi_x(n) := e^{in\cdot x}.
\end{equation}
While $\phi_x \not\in \ell^2(\Z^d)$, $h_0$ acts on it via
(\ref{eq:Laplace}) as
\begin{eqnarray*}
(h_0 \phi_x)(n) & = & - \sum_{|k|=1} e^{i(n+k)\cdot x} \\
& = & \left( - \sum_{j=1}^d 2\cos(x_j) \right) \phi_x(n).
\end{eqnarray*}
Thus $\phi_x$ is a bounded generalized eigenfunction of $h_0$ to the
spectral value $-2\sum_j \cos(x_j)$.

\subsection{The Anderson Model}

Let $\omega = (\omega_n)_{n\in \Z^d}$ be a set of independent,
identically  distributed (i.i.d.) real-valued random variables
indexed by $n\in \Z^d$. Recall that this means the following, where
we denote probabilities by $\PP$:

\begin{itemize}

\item The $(\omega_n)$ are identically distributed, i.e.\ there
exists a Borel probability measure $\mu$ on $\R$ such that, for all
$n\in \Z^d$ and Borel sets $A\subset \R$,
\[ \PP(\omega_n \in A) = \mu(A).\]

\item The $(\omega_n)$ are independent. Thus, for each finite subset $\{n_1,\ldots, n_{\ell}\}$ of $\Z^d$ and arbitrary Borel sets $A_1, \ldots, A_{\ell} \subset \R$,
\begin{eqnarray*}
\PP(\omega_{n_1} \in A_1, \ldots, \omega_{n_{\ell}} \in A_{\ell}) &
= & \prod_{j=1}^{\ell} \PP(\omega_{n_j} \in A_j) \\ & = &
\prod_{j=1}^{\ell} \mu(A_j).
\end{eqnarray*}

\end{itemize}

It is sometimes useful to think of a concrete way in which i.i.d.\
random variables can be realized as measurable functions on a
probability space $(\Omega, {\mathcal A}, \PP)$. The standard
construction is the infinite product space
\[ (\Omega, {\mathcal A}, \PP) = \bigotimes_{n\in \Z^d} (\R, {\mathcal B}_{\R}, \mu),\]
with ${\mathcal A}$ and $\PP$ denoting the $\sigma$-algebra and
measure generated by the pre-measure induced by $\mu$ on the Borel
cylinder sets in $\Omega = \R^{\Z^d}$. This is consistent with the
notation $\omega = (\omega_n)_{n\in \Z^d}$ as the components
$\omega_n$ of $\omega\in \Omega$ are now i.i.d.\ random variables on
$\Omega$ with common distribution $\mu$.

It is also convenient to work on a complete probability space
$(\Omega, {\mathcal A}, \PP)$, which in the above realization is
achieved by completing the product algebra ${\mathcal A}$ under
$\PP$, for which the same notation will be kept.

The {\it Anderson Model} is a random Hamiltonian $h_{\omega}$ on
$\ell^2(\Z^d)$, defined for $\omega \in \Omega$ by
\begin{equation} \label{eq:Anderson}
(h_{\omega} u)(n) = (h_0 u)(n) +  \omega_n u(n), \quad n\in \Z^d.
\end{equation}

Introducing the {\it random potential} $V_{\omega}: \Z^d \to \R$ by
$V_{\omega}(n) = \omega_n$, we may also write
\[ h_{\omega} = h_0 + V_{\omega}.\]

Note here that $h_{\omega}$ is not a single operator, but rather an
operator-valued function on a probability space. It's operator
theoretic properties will generally depend on $\omega$. Our goal
will typically be that a certain property of $h_{\omega}$ holds {\it
almost surely} or {\it with probability one}, meaning that it holds
for $\omega\in \Omega_0$, a measurable subset of $\Omega$ with
$\PP(\Omega_0)=1$. It lies within the nature of random operator
theory that the most interesting properties will only hold almost
surely rather than for {\it all} $\omega \in \Omega$.

One may think of the Anderson model $h_{\omega}$ as the Hamiltonian
governing the quantum mechanical motion of a single electron in a
discretized alloy-type random medium. In this view the random
potential $V_{\omega}(n) = \omega_n$, $n\in \Z^d$, represents a
solid formed by nuclei located at the sites $n$ of the lattice
$\Z^d$ and carrying random electrical charges $ \omega_n$. Assuming
that $h_{\omega}$ is self-adjoint, the dynamics of the electron is
given through the unitary group $e^{-ith_{\omega}}$, defined via the
spectral theorem, which provides the solution $\psi(t) =
e^{-ith_{\omega}} \psi_0$ of the time-dependent Schr\"odinger
equation $h_{\omega} \psi(t) = i \psi'(t)$, $\psi(0)=\psi_0$. The
possible energies of the electron are given by the spectrum
$\sigma(h_{\omega})$ of the Anderson Hamiltonian $h_{\omega}$.

In the above discussion we have assumed self-adjointness of the
Hamiltonian, which is the first mathematical fact to be checked.
This is particularly easy for discrete Schr\"odinger operators such
as $h_{\omega}$ because the discrete Laplacian $h_0$ is bounded and
self-adjoint.

\begin{theorem} \label{thm:selfadjoint}
For every $\omega \in {\R}^{\Z^d}$, the operator $h_{\omega}$ is
self-adjoint on
\[ D(V_{\omega}) = \{ u\in \ell^2(\Z^d): \sum_n |\omega_n u(n)|^2 < \infty\},\]
the domain of the maximal multiplication operator by the potential
$V_{\omega}$.
\end{theorem}

\begin{proof}
Self-adjointness of the maximal multiplication operator by a
real-valued function is a standard fact. Perturbation of the
self-adjoint maximal multiplication operator $V_{\omega}$ by the
bounded self-adjoint operator $h_0$ preserves self-adjointness with
same domain, e.g.\ \cite{Weidmann}.
\end{proof}

If we assume, as will be done later, that the distribution $\mu$ of
the $\omega_n$ has bounded support, i.e.\ that
\[ \mbox{supp}\,\mu := \{ t \in \R:\, \mu((t-\varepsilon, t+\varepsilon))>0 \:\mbox{for all $\varepsilon >0$}\}\]
is bounded in $\R$, then the potential $V_{\omega}$ is bounded and
therefore defines a bounded multiplication operator. Thus
$h_{\omega}$ is a bounded self-adjoint operator on $\ell^2(\Z^d)$ as
well. On the other hand, if supp$\,\mu$ is unbounded, then it is not
hard to see that $V_{\omega}$, and thus $h_{\omega}$, is almost
surely unbounded.

\subsection{The spectrum of the Anderson model}

Our next goal is to determine the spectrum of $h_{\omega}$. It
follows  as a consequence of the general theory of so-called {\it
ergodic operators} (e.g.\ \cite{Carmona/Lacroix}), of which the
Anderson model is a special case, that $\sigma(h_{\omega})$ is
almost surely deterministic, i.e.\ there exists a closed subset
$\Sigma$ of $\R$ such that
\[ \sigma(h_{\omega}) = \Sigma \quad \mbox{almost surely.} \]

Rather than proving this within the general theory of ergodic
operators  we will give a direct proof of the following result,
which explicitly describes the almost sure spectrum of the Anderson
model:

\begin{theorem} \label{thm:charspectrum}
The spectrum of the Anderson model is almost surely given by
\begin{equation} \label{eq:charspectrum}
\sigma(h_{\omega}) = \sigma(h_0) + \,\mbox{supp}\,\mu.
\end{equation}
\end{theorem}

Here the sum of two subsets $A$ and $B$ of $\R$ is defined by
$A+B:=\{a+b: a\in A, b\in B\}$. In particular, this means that the
almost sure spectrum of $h_{\omega}$ is a union of intervals, namely
of translates of $[-2d,2d]$ by the points in supp$\,\mu$. If
supp$\,\mu$ doesn't have large gaps, then the almost sure spectrum
of $h_{\omega}$ is a single interval.

\begin{proof}
We begin with the easy part of the proof, namely that
\begin{equation} \label{eq:firsthalf}
\sigma(h_{\omega}) \subset [-2d,2d]+ \mbox{supp}\,\mu
\end{equation}
almost surely.

We first argue that $\sigma(V_{\omega}) = \overline{\{ \omega_n:
n\in \Z^d\}} \subset \mbox{supp}\,\mu$ almost surely. In fact, as
$\mu(\mbox{supp}\,\mu)=1$, for fixed $n\in \Z^d$, $\omega_n \in
\mbox{supp}\,\mu$ holds almost surely, i.e.\ on a set $\Omega_n
\subset \Omega$ with $\PP(\Omega_n)=1$. The countable intersection
$\Omega'$ of the $\Omega_n$ also has measure one and for $\omega\in
\Omega'$ we have $\sigma(V_{\omega}) \subset \mbox{supp}\,\mu$ as
supp$\,\mu$ is closed.

By a general fact from spectral theory, easily proven using a
Neumann series argument, a bounded self-adjoint perturbation $B$
does not shift the spectrum of a self-adjoint operator $A$ by more
than $\|B\|$, i.e.\ $\sigma(A+B) \subset \sigma(A) + [-\|B\|,
\|B\|]$. Thus (\ref{eq:firsthalf}) holds for $\omega\in \Omega'$ by
(\ref{eq:Laplacespec}).

The proof of
\begin{equation} \label{eq:secondhalf}
 [-2d,2d] + \mbox{supp}\,\mu \subset \sigma(h_{\omega})
\end{equation}
with probability one is more involved and falls into a probabilistic
part and a spectral theoretic part.

For the probabilistic part, start with fixed $t \in
\mbox{supp}\,\mu$, $\varepsilon>0$ and $N\in \N$ and let
\begin{eqnarray*}
\Omega_{t,N,\varepsilon} & := & \left\{ \omega\in \Omega: \;\mbox{There exists a cube $\Lambda_N \subset \Z^d$ of side length $N$ such that} \right. \\
& & \left. |\omega_n -t| < \varepsilon \mbox{ for all $n\in
\Lambda_N$}\right\}.
\end{eqnarray*}
As $t \in \mbox{supp}\,\mu$, we have $p:=
\mu((t-\varepsilon,t+\varepsilon))>0$.  Thus, for each fixed cube
$\Lambda$ of side length $N$ in $\Z^d$, the probability that
$|\omega_n-t|<\varepsilon$ for all $n\in \Lambda$ is $p^{N^d}
>0$. We can cover $\Z^d$ by infinitely many disjoint cubes of side
length $N$, where these events are independent. It follows that
$\PP(\Omega_{t,N,\varepsilon})=1$.

Next, let $\Omega_{t,\varepsilon} := \cap_{N\in \N}
\Omega_{t,N,\varepsilon}$.  Thus $\PP(\Omega_{t,\varepsilon})=1$ and
we will prove the following below: For each $\omega \in
\Omega_{t,\varepsilon}$ and $a\in [-2d,2d]$ it holds that
\begin{equation} \label{eq:specfact}
[a+t-\varepsilon, a+t+\varepsilon] \cap \sigma(h_{\omega}) \not=
\emptyset.
\end{equation}

Assuming that (\ref{eq:specfact}) is true, we proceed as follows:
Define $\Omega_{t} := \cap_{\ell \in \N} \Omega_{t,1/\ell}$, such
that $\PP(\Omega_{t})=1$. For $\omega \in \Omega_{t}$ we have by
(\ref{eq:specfact}) that, for all $a\in [-2d,2d]$,
\[ (a+t-\frac{1}{\ell}, a+t +\frac{1}{\ell}) \cap \sigma(h_{\omega}) \not= \emptyset \]
for all $\ell \in \N$. As $\sigma(h_{\omega})$ is closed, this
implies  that $a+t \in \sigma(h_{\omega})$ for all $a\in [-2d,2d]$,
and thus $[-2d,2d] +t \subset \sigma(h_{\omega})$.

For one last argument involving countable intersections of full
measure  sets, let $B$ be a countable subset of supp$\,\mu$ which is
dense in supp$\,\mu$ and let $\Omega_0 := \cap_{t\in B} \Omega_{t}$.
Then $\PP(\Omega_0)=1$ and for $\omega \in \Omega_0$ we have
$[-2d,2d]+B \subset \sigma(h_{\omega})$. Using again that
$\sigma(h_{\omega})$ is closed completes the proof of
(\ref{eq:secondhalf}).

We still need to show (\ref{eq:specfact}), which is the
spectral-theoretic  part of the proof. Let $\omega \in
\Omega_{t,\varepsilon}$. Thus, by assumption, for each $N\in\N$
there exists a cube $\Lambda_N$ of side length $N$ such that
$|\omega_n-t| < \varepsilon$ for all $n\in \Lambda_N$.

To $a\in [-2d,2d]$ pick $x=(x_1,\ldots,x_d)$ such that $a=-2\sum_j
\cos(x_j)$ and consider the corresponding generalized eigenfunction
$\phi_x$ from (\ref{eq:genef}). Then $\psi_N := \chi_{\Lambda_N}
\phi_x$ has finite support and, in particular, lies in $ \ell^2
(\Z^d)$. We claim that
\begin{equation} \label{eq:weyl}
\limsup_{N\to\infty} \frac{\|(h_{\omega}-(a+t))\psi_N\|}{\|\psi_N\|}
\le \varepsilon.
\end{equation}

To find a norm bound for $(h_{\omega}-(a+t))\psi_N =
(h_0-a)\chi_{\Lambda_N} \phi_x +(V_{\omega}-t) \chi_{\Lambda_N}
\phi_x$, we first note that by assumption
$\|(V_{\omega}-t)\chi_{\Lambda_N} \phi_x\| \le \varepsilon
\|\psi_N\|$. Moreover, as $(h_0-a)\phi_x =0$, it follows that
$((h_0-a)\chi_{\Lambda_N} \phi_x)(n)$ is non-zero only for $n$ close
to the boundary of $\Lambda_N$, where its values are bounded by a
constant independent of $N$. Thus
\[ \|(h_{\omega}-(a+t))\psi_N\| \le CN^{(d-1)/2} + \varepsilon \|\psi_N\|.\]
On the other hand we have $\|\psi_N\| = N^{d/2}$. This proves
(\ref{eq:weyl}).

We conclude by a standard argument: If $h_{\omega}-(a+t)$ is
invertible, then, by (\ref{eq:weyl}),
\[ \|(h_{\omega}-(a+t))^{-1}\| \ge \frac{1}{\varepsilon}.\]
This implies (\ref{eq:specfact}) by using the fact that for general
self-adjoint operators $A$ it holds that
\[ \|(A-z)^{-1}\| = \frac{1}{\mbox{dist}(z,\sigma(A))}.\]

\end{proof}

\section{Localization Properties} \label{sec:locproperties}

We will be interested in {\it localization properties} of the
Anderson model, which can be described either by spectral properties
or by dynamical properties of the Hamiltonian.

To be more precise, let $I\subset \R$ be an open interval. We say
that $h_{\omega}$ exhibits {\it spectral localization} in $I$ if
$h_{\omega}$ almost surely has pure point spectrum in $I$, i.e.\ $I$
does not contain any continuous spectrum of $h_{\omega}$, and its
eigenfunctions to all eigenvalues in $I$ decay exponentially.

If $I$ is a non-trivial interval contained in the almost sure
spectrum of $h_{\omega}$, which is a union of intervals, then
spectral localization in $I$ necessarily means that the spectrum
consists of a dense set of eigenvalues (whose closure fills all of
$I$). This phenomenon is very different and much more subtle than
the appearance of {\it discrete} isolated eigenvalues, which is the
classical situation encountered in atomic or molecular hamiltonians.
In fact, the possibility of dense pure point spectrum historically
can be considered as the biggest mathematical surprise provided by
the investigation of the Anderson model.

On the other hand, we say that $h_{\omega}$ exhibits {\it dynamical
localization} in $I$ if there exist constants $C<\infty$ and
$\mu>0$ such that
\begin{equation} \label{eq:dynloc}
\E \left( \sup_{t\in \R} | \langle e_j, e^{-ith_{\omega}}
\chi_I(h_{\omega}) e_k \rangle | \right) \le C e^{-\mu |j-k|},
\end{equation}
for all $j, k \in \Z^d$. Here $\{e_j\}_{j\in \Z^d}$ is the canonical
orthonormal basis in $\Z^d$, $e_j(k) = \delta_{jk}$, and $\E(\cdot)$
denotes the expectation with respect to the probability measure
$\PP$, meaning $\E(X) = \int_{\Omega} X \,d\PP$ for random variables
$X$ on $\Omega$. Both, $e^{-ith_{\omega}}$ as well as
$\chi_I(h_{\omega})$, are defined via the functional calculus for
self-adjoint operators. By $\chi_I$ we denote the characteristic
function of $I$, so that $\chi_I(h_{\omega})$ is the spectral
projection for $h_{\omega}$ onto $I$.

Dynamical localization in the form (\ref{eq:dynloc}) is a strong
form of asserting that solutions of the time-dependent Schr\"odinger
equation $h_{\omega} \psi(t) = i\partial_t \psi(t)$ are staying
localized in space, uniformly for all times, and thus shows the
absence of {\it quantum transport}. Let us illustrate this by
showing that dynamical localization implies that all moments of the
position operator are bounded in time, i.e.\ for all $p>0$ and all
finitely supported $\psi \in \ell^2(\Z^d)$,
\begin{equation} \label{eq:moments}
\sup_{t\in \R} \| |X|^p e^{-ith_{\omega}} \chi_I(h_{\omega}) \psi\|
< \infty \quad \mbox{almost surely},
\end{equation}
where the position operator $|X|$ is defined by $(|X|\phi)(n) =
|n|\phi(n)$. To see how (\ref{eq:moments}) follows from
(\ref{eq:dynloc}), assume that $\psi(k)=0$ for $|k|>R$.  Then
\begin{eqnarray*}
\||X|^p e^{-ith_{\omega}} \chi_I(h_{\omega}) \psi\|^2 & = & \sum_j
\left| \langle e_j, |X|^p e^{-ith_{\omega}} \chi_I(h_{\omega}) \psi
\rangle \right|^2 \\
& = & \sum_j |j|^{2p} \left| \sum_{|k|\le R} \langle e_j,
e^{-ith_{\omega}} \chi_I(h_{\omega}) e_k \rangle \psi(k) \right|^2
\\
& \le & \sum_j \sum_{|k|\le R} |j|^{2p} \left| \langle e_j,
e^{-ith_{\omega}} \chi_I(h_{\omega}) e_k \rangle \right|^2
\|\psi\|^2,
\end{eqnarray*}
where the last step used the Cauchy-Schwarz inequality. We can drop
the square from $|\langle e_j, e^{-ith_{\omega}} \chi_I(h_{\omega})
e_k \rangle |^2$ (as this number is bounded by $1$) and then take
expectations to get
\begin{eqnarray*}
\lefteqn{\E \left( \sup_t \||X|^p e^{-ith_{\omega}} \chi_I(h_{\omega})
\psi\|^2 \right)} \\ & \le & \sum_j \sum_{|k|\le R} |j|^{2p} \E \left(
\sup_t |\langle e_j, e^{-ith_{\omega}} \chi_I(h_{\omega}) e_k
\rangle | \right) \|\psi\|^2 \\
& \le & C \sum_j \sum_{|k|\le R} |j|^{2p} e^{-\mu|j-k|} \|\psi\|^2
\\
& < & \infty.
\end{eqnarray*}
This implies the almost sure statement in (\ref{eq:moments}) (with
square at the norm and therefore also without).

Dynamical localization is not only the physically more interesting
statement than spectral localization (as physicists usually have
little patience and limited appreciation for spectral theory), it is
also the mathematically stronger property: We will show later that
dynamical localization in $I$ implies spectral localization in $I$.

Let us discuss situations in which localization, spectral or
dynamical, is expected physically. For this it will help to
introduce an additional disorder parameter $\lambda>0$ in the
Anderson model and define
\begin{equation} \label{eq:Andersondis}
h_{\omega,\lambda} = h_0 + \lambda V_{\omega},
\end{equation}
with $V_{\omega}(n)= \omega_n$ as above. Formally, this fits into
the same framework as (\ref{eq:Anderson}), using the re-scaled
distribution
\begin{equation} \label{eq:scaling}
\PP(\lambda \omega_n \in B) = \mu_{\lambda}(B) := \mu(B/\lambda)
\end{equation}
 of the i.i.d.\ random variables $\lambda \omega_n$.
The distribution $\mu_{\lambda}$ is spread out over larger supports
for larger $\lambda$, corresponding to a wider range of possible
random charges in an alloy-type medium. Thus $\lambda
>>1$ is the case of {\it large disorder} and $\lambda << 1$
represents {\it small disorder}.

Physicists know all of the following:

In dimension $d=1$ the entire spectrum of $h_{\omega,\lambda}$ is
localized for {\it any} value of the disorder $\lambda>0$.

In dimension $d\ge 2$ the entire spectrum is localized at {\it large
disorder}, i.e.\ for $\lambda>>1$.

For small disorder $\lambda$ different behavior arises in dimensions
$d=2$ and $d=3$. For $d=2$ one still has localization of the entire
spectrum, but possibly in a weaker form than for $d=1$, e.g.\ a
small amount (or weak type) of quantum transport might be possible.
On the other hand, in dimension $d=3$ one observes the {\it Anderson
transition}. There are localized regions near the band edges of the
almost sure spectrum, separated by {\it mobility edges} from a
region of {\it extended states} in the center of the spectrum.
Extended states are interpreted as the existence of quantum
transport in the sense that the moments (\ref{eq:moments}) should be
infinite for sufficiently large $p$. The physical expectation for
$d=3$ is that this starts at $p=1/2$, which corresponds to the
presence of diffusive motion.

Mathematically, localization has been proven for three different
regimes: (i) for all energies and arbitrary disorder in $d=1$, (ii)
in any dimension and for all energies at sufficiently large
disorder, and (iii) near band edges of the spectrum in any dimension
and for arbitrary disorder.

The mechanisms which cause localization in the Anderson model are
fundamentally different for the one-dimensional and
multi-dimensional case, which is also reflected in the mathematical
methods which have been used to prove this. In $d=1$ strong tools
from the theory of one-dimensional dynamical systems are available,
in particular results on the asymptotics of products of independent
random variables which allow to prove positivity of Lyapunov
exponents.  Large parts of the books \cite{Carmona/Lacroix} and
\cite{Pastur/Figotin} are devoted to the presentation of the
one-dimensional theory. A complete presentation of the
Kunz-Souillard proof of localization for the one-dimensional
Anderson model can be found in \cite{CFKS}. For a somewhat later
survey of results on one-dimensional localization see
\cite{StolzIndia}.

As discussed in the introduction, we will focus here on methods
which allow to prove multi-dimensional localization and, among the
two methods which have been shown to accomplish this, focus on the
fractional moments method. Using this method we will give a detailed
proof of large disorder localization and also explain how it works
to show band edge localization, in each case in arbitrary dimension.

We will not discuss localization proofs via multiscale analysis.
Excellent introductions to this method can be found in \cite{Kirsch}
and \cite{Stollmann}, while the state of the art of what can be
obtained from Fr\"ohlich-Spencer-type multiscale analysis is
presented in \cite{Germinet/Klein} and the review \cite{Klein}. We
also mention the recent powerful extension of the ideas behind
multiscale analysis in \cite{Bourgain/Kenig}, which allow to prove
localization for continuum Anderson models (see
Section~\ref{sec:continuum}) with discretely distributed random
couplings, a result which is beyond what can be obtained by the
fractional moments method.

\section{Localization at large disorder} \label{sec:largedisorder}

Consider the Anderson model (\ref{eq:Andersondis}) at disorder
$\lambda>0$ and in any dimension $d\ge 1$.

Throughout the rest of these notes we will work with a stronger
assumption on the distribution $\mu$ of the random parameters
$\omega_n$, namely that $\mu$ is absolutely continuous with density
$\rho$,
\begin{equation} \label{eq:density}
\mu(B) = \int_B \rho(v)\,dv \quad \mbox{for $B\subset \R$\;Borel},
\quad \rho \in L^{\infty}_0(\R),
\end{equation}
i.e.\ $\rho$ is bounded and has compact support. In particular, this
means that the Anderson hamiltonian $h_{\omega,\lambda}$ is a
bounded self-adjoint operator in $\ell^2(\Z^d)$.

Introduce the {\it Green function} as the matrix-elements of the
resolvent of $h_{\omega,\lambda}$,
\begin{equation} \label{eq:green}
G_{\omega,\lambda}(x,y;z) := \langle e_x,
(h_{\omega,\lambda}-z)^{-1} e_y \rangle.
\end{equation}

Our first goal is to prove

\begin{theorem}[\cite{Aizenman/Molchanov}] \label{thm1}
Let $0<s<1$. Then there exists $\lambda_0>0$ such that for $\lambda
\ge \lambda_0$ there are $C<\infty$ and $\mu>0$ with
\begin{equation} \label{eq:fmdecay}
\E \left( |G_{\omega,\lambda}(x,y;z)|^s \right) \le Ce^{-\mu |x-y|}
\end{equation}
uniformly in $x,y \in \Z^d$ and $z\in \C\setminus \R$.
\end{theorem}

It is the appearance of fractional moments of the form
$\E(|\cdot|^s)$, $0<s<1$, in the above theorem which prompted the
name ``fractional moments method'' for the circle of ideas which we
want to present here. The method is also frequently called the
``Aizenman-Molchanov method'', as Aizenman and Molchanov did not
only realize that results such as Theorem~\ref{thm1} hold, but that
they imply spectral and dynamical localization. These implications
will be discussed in the next section.

The following proof of Theorem~\ref{thm1} closely follows the
original ideas from \cite{Aizenman/Molchanov}. We start with two
lemmas, an {\it a-priori bound} on the fractional moments of Green's
function and a {\it decoupling lemma}, which contain central ideas
behind the method and, in increasing degree of sophistication, have
been used in all subsequent developments of the method.

\begin{lemma}[A priori bound] \label{lem1}
There exists a constant $C_1 = C_1(s,\rho)<\infty$ such that
\begin{equation} \label{eq:apriori}
\E_{x,y}(|G_{\omega,\lambda}(x,y;z)|^s) \le C_1 \lambda^{-s}
\end{equation}
for all $x,y \in \Z^d$, $z \in \C\setminus \R$, and $\lambda>0$.
\end{lemma}

Here
\[
\E_{x,y}(\ldots) = \int \int \ldots \rho(\omega_x)\,d\omega_x
\,\rho(\omega_y)\,d\omega_y
\]
is the conditional expectation with $(\omega_u)_{u\in \Z^d \setminus
\{x,y\}}$ fixed. After averaging over $\omega_x$ and $\omega_y$ the
bound in (\ref{eq:apriori}) does not depend on the remaining random
parameters. Thus we get also that
\[
\E(|G_{\omega,\lambda}(x,y;z)|^s) \le C_1 \lambda^{-s}.
\]

\begin{proof}
We first prove (\ref{eq:apriori}) for the case $x=y$, which
demonstrates the simplicity of the fundamental idea underlying the
FMM. For fixed $x\in \Z^d$, write $\omega = (\hat{\omega},
\omega_x)$ where $\hat{\omega}$ is short for $(\omega_u)_{u \in \Z^d
\setminus \{x\}}$. With $P_{e_x} := \langle e_x, \cdot \rangle e_x$,
the orthogonal projection onto the span of $e_x$, we can separate
the $\omega_x$ and $\hat{\omega}$ dependence of $h_{\omega,\lambda}$
as
\[
h_{\omega,\lambda} = h_{\hat{\omega},\lambda} + \lambda \omega_x
P_{e_x}.
\]
The resolvent identity yields
\begin{equation} \label{eq:resid}
(h_{\omega,\lambda}-z)^{-1} = (h_{\hat{\omega},\lambda}-z)^{-1}
-\lambda \omega_x (h_{\hat{\omega},\lambda}-z)^{-1} P_{e_x}
(h_{\omega,\lambda}-z)^{-1}.
\end{equation}
Taking matrix-elements we conclude for the corresponding diagonal
Green functions that
\begin{equation} \label{eq:diaggreen}
G_{\omega,\lambda}(x,x;z) = G_{\hat{\omega},\lambda}(x,x;z) -
\lambda \omega_x G_{\hat{\omega},\lambda}(x,x;z)
G_{\omega,\lambda}(x,x;z)
\end{equation}
or
\begin{equation} \label{eq:solveG}
G_{\omega,\lambda}(x,x;z) = \frac{1}{a+\lambda \omega_x} \quad
\mbox{with $a= \frac{1}{G_{\hat{\omega},\lambda}(x,x;z)}$}.
\end{equation}
Note that the latter is well-defined since one can easily check the
Herglotz property
$\mbox{Im}\,G_{\hat{\omega},\lambda}(x,x;z)/\mbox{Im}\,z
>0$ of the Green function.

The important fact is that $a$ is a complex number which does not
depend on $\omega_x$. Thus, writing $\E_x(\ldots) := \int \ldots
\rho(\omega_x)\,d\omega_x$, we find that
\begin{equation} \label{eq:fracint}
\E_x(|G_{\omega,\lambda}(x,x;z)|^s) \le
\frac{\|\rho\|_{\infty}}{\lambda^s} \int_{\mbox{supp}\,\rho}
\frac{d\omega_x}{|\frac{a}{\lambda} +\omega_x|^s} \le
\frac{C(\rho,s)}{\lambda^s},
\end{equation}
with $C(\rho,s)$ independent of $\lambda$ and $a$, and thus
independent of $\hat{\omega}$, $z$ and $x$.

The proof of (\ref{eq:apriori}) for $x\not= y$ is based on the same
idea, replacing the rank-one-perturbation arguments above with
rank-two-perturbation arguments. Write $\omega = (\hat{\omega},
\omega_x, \omega_y)$, $P=P_{e_x}+P_{e_y}$ and
\[
h_{\omega,\lambda} = h_{\hat{\omega},\lambda} + \lambda \omega_x
P_{e_x} + \lambda \omega_y P_{e_y}.
\]
Using the resolvent identity similar to above one arrives at
\begin{equation} \label{eq:krein}
P(h_{\omega,\lambda}-z)^{-1} P = \left( A+\lambda \left(
\begin{array}{cc} \omega_x & 0 \\ 0 & \omega_y \end{array} \right)
\right)^{-1},
\end{equation}
where
\[
A = (P(h_{\hat{\omega},\lambda}-z)^{-1} P)^{-1},
\]
both to be read as identities for $2\times 2$-matrices in the range
of $P$. This is a special case of the {\it Krein formula} which
characterizes the resolvents of finite-rank perturbations of general
self adjoint operators. For the matrix $A$ one can check that Im$\,A
= \frac{1}{2i}(A-A^*)<0$ if Im$\,z>0$ and Im$\,A>0$ if Im$\,z<0$. It
is also independent of $\omega_x$ and $\omega_y$.

Using that $G_{\omega,\lambda}(x,y;z)$ is one of the matrix-elements
of $P(h_{\omega,\lambda}-z)^{-1} P$, we find
\begin{eqnarray*}
\E_{x,y}(|G_{\omega,\lambda}(x,y;z)|^s) & \le &
\E_{x,y}\left(\left\|\left(A+ \lambda \left(
\begin{array}{cc} \omega_x & 0 \\ 0 & \omega_y \end{array}
\right)\right)^{-1}\right\|^s\right) \\
& = & \lambda^{-s} \E_{x,y}\left(\left\|\left(-\frac{1}{\lambda}A
-\left(
\begin{array}{cc} \omega_x & 0 \\ 0 & \omega_y \end{array}
\right)\right)^{-1}\right\|^s\right)\\
& \le & \frac{\|\rho\|^2_{\infty}}{\lambda^s} \int_{-r}^r
\int_{-r}^r
\left\| \left( -\frac{1}{\lambda}A - \left( \begin{array}{cc} \omega_x & 0 \\
0 & \omega_y \\ \end{array} \right) \right)^{-1}
\right\|^s\,d\omega_x\, d\omega_y,
\end{eqnarray*}
where $[-r,r]$ is an interval containing supp$\,\rho$. In the double
integral we change variables to
\[
u=\frac{1}{2}(\omega_x+\omega_y), \quad
v=\frac{1}{2}(\omega_x-\omega_y),
\]
which gives a Jacobian factor of $2$. As $(\omega_x,\omega_y)\in
[-r,r]^2$ implies $(u,v)\in [-r,r]^2$ we arrive at the bound
\begin{eqnarray*}
\lefteqn{\E_{x,y}(|G_{\omega,\lambda}(x,y;z)|^s)} \\
& \le & \frac{2\|\rho\|_{\infty}^2}{\lambda^s} \int_{-r}^r
\int_{-r}^r \left\| \left( -\frac{1}{\lambda} A + \left(
\begin{array}{cc} -v & 0 \\ 0 & v \end{array} \right)
-uI\right)^{-1} \right\|^s \,du\,dv \\
& \le & \frac{4r\|\rho\|^2_{\infty}}{\lambda^s} C(r,s) =
\frac{C(s,\rho)}{\lambda^s}.
\end{eqnarray*}
That the latter bound is uniform in $x$, $y$ and $z$ as well as in
the random parameters $(\omega_u)_{u\in \Z^d\setminus \{x,y\}}$
follows from the fact that the matrix
\[ -\frac{1}{\lambda} A + \left( \begin{array}{cc} -v & 0 \\ 0 & v
\end{array} \right) \]
has either positive or negative imaginary part and the following
general result:

For every $s\in (0,1)$ and $r>0$ there exists $C(r,s)<\infty$ such
that
\begin{equation} \label{eq:2by2fracbound}
\int_{-r}^r \|(B-uI)^{-1}\|^s\,du \le C(r,s)
\end{equation}
for all $2\times 2$-matrices $B$ such that either Im$\,B \ge 0$ or
Im$\,B\le 0$.

Let us reproduce an elementary proof of this fact, e.g.\ Lemma~4.1
in \cite{HJS},  starting with the observation that, by Schur's
Theorem, $B$ may be assumed upper triangular. We also may assume
without loss that Im$\,B \ge 0$.

Thus
\begin{equation} \label{eq:Aform}
B = \left( \begin{array}{cc} b_{11} & b_{12} \\ 0 & b_{22}
\end{array} \right)
\end{equation}
and
\begin{equation} \label{eq:disinvers}
(B-uI)^{-1} = \left( \begin{array}{cc} \frac{1}{b_{11}-u} & -
\frac{b_{12}}{(b_{11}-u)(b_{22}-u)} \\ 0 & \frac{1}{b_{22}-u}
\end{array} \right).
\end{equation}
The bound (\ref{eq:2by2fracbound}) follows if we can establish a
corresponding fractional integral bound for the absolute value of
each entry of (\ref{eq:disinvers}) separately. For the diagonal
entries this is obvious.

We bound the upper right entry of (\ref{eq:disinvers}) by
\begin{eqnarray} \label{eq:upperright}
\left| \frac{b_{12}}{(b_{11}-u)(b_{22}-u)}\right| & \le & \frac{|b_{12}|}{|\mbox{Im}\,((b_{11}-u)(b_{22}-u))|} \nonumber \\
& = & \frac{1}{\left| u \frac{\mbox{\footnotesize Im}\,b_{11}
+\mbox{\footnotesize Im}\,b_{22}}{|b_{12}|}
-\frac{\mbox{\footnotesize Im}(b_{11}b_{22})}{|b_{12}|} \right|}.
\end{eqnarray}
The positive matrix
\[
\mbox{Im}\,B = \left( \begin{array}{cc} \mbox{Im}\,b_{11} &
\frac{1}{2i}b_{12} \\ -\frac{1}{2i}\bar{b}_{12} & \mbox{Im}\,b_{22}
\end{array} \right)
\]
has positive determinant, i.e.\ det Im$\,B = \mbox{Im}\,b_{11}
\mbox{Im}\,b_{22} -|b_{12}|^2/4$. We thus get
\[
\left| \frac{\mbox{Im}\,b_{11} + \mbox{Im}\,b_{22}}{b_{12}}\right|^2
\ge \frac{2\mbox{Im}\,b_{11} \mbox{Im}\,b_{22}}{|b_{12}|^2} \ge
\frac{1}{2}.
\]
The latter allows to conclude the required integral bound for
(\ref{eq:upperright}).

\end{proof}

The other result needed for the proof of Theorem~\ref{thm1} is

\begin{lemma}[Decoupling Lemma] \label{lem2}
For a compactly supported and bounded density function $\rho$ as
above there exists a constant $C_2<\infty$ such that
\begin{equation} \label{eq:decoupling}
\frac{\int \frac{1}{|v-\beta|^s} \rho(v)\,dv}{\int
\frac{|v-\eta|^s}{|v-\beta|^s} \rho(v)\,dv} \le C_2
\end{equation}
uniformly in $\eta, \beta \in \C$.
\end{lemma}

This can be understood as a consequence of the following two facts:
(i) The two integrals on the left hand side of (\ref{eq:decoupling})
are continuous functions of $\eta$ and $\beta$. As both of them
neither vanish nor diverge, the same is true for the ratio of the
integrals. (ii) As $|\beta|$ and $|\eta|$ become large, the left
hand side of (\ref{eq:decoupling}) has finite limits. This combines
to give a uniform bound in $\beta$ and $\eta$. The details are left
as an exercise, or can be found in \cite{Graf94}.

We are now prepared to complete the proof of Theorem~\ref{thm1}:

\begin{proof}
Given the a-priori bound from Lemma~\ref{lem1} we may assume $y\not=
x$. Then
\begin{eqnarray} \label{eq:expansion}
0 & = & \langle e_x, e_y \rangle \\ \nonumber & = & \langle e_x,
(h_{\omega,\lambda}-z)^{-1} (h_{\omega,\lambda}-z) e_y \rangle \\
\nonumber & = & \Big\langle e_x, (h_{\omega,\lambda}-z)^{-1} \Big(-
\sum_{u:|u-y|=1} e_u + (\lambda \omega_y-z)e_y \Big) \Big\rangle \\
\nonumber & = & - \sum_{u:|u-y|=1} G_{\omega,\lambda}(x,u;z) +
(\lambda \omega_y-z) G_{\omega,\lambda}(x,y;z).
\end{eqnarray}

Note that $G_{\omega,\lambda}(x,y;z)$ is the upper left entry  of
the matrix on the left hand side of the Krein formula
(\ref{eq:krein}). Explicitly inverting the right hand side of
(\ref{eq:krein}) we find that
\[ G_{\omega,\lambda}(x,y;z) = \frac{\alpha}{\lambda \omega_y -\beta},\]
where $\alpha$ and $\beta$ do not depend on $\omega_y$ (and it will
not matter that they depend on $\lambda$). Using Lemma~\ref{lem2},
the bound $(\sum_n |a_n|)^s \le \sum_n |a_n|^s$ and
(\ref{eq:expansion}) we find
\begin{eqnarray} \label{eq:itstep}
\E(|G_{\omega,\lambda}(x,y;z)|^s) & = & \frac{1}{\lambda^s} \E \Big(
\Big| \frac{\alpha}{\omega_y-\frac{\beta}{\lambda}} \Big|^s \Big) \\
\nonumber & \le & \frac{C_2}{\lambda^s} \E \Big( |\alpha|^s
\frac{|\omega_y-\frac{z}{\lambda}|^s}{|\omega_y -
\frac{\beta}{\lambda}|^s} \Big) \\ \nonumber & = &
\frac{C_2}{\lambda^s} \E (|\lambda \omega_y -z|^s
|G_{\omega,\lambda}(x,y;z)|^s) \\ \nonumber & \le &
\frac{C_2}{\lambda^s} \sum_{u:|u-y|=1}
\E(|G_{\omega,\lambda}(x,u;z)|^s).
\end{eqnarray}
If none of the lattice sites $u$ are equal to $x$, then the argument
can be iterated. For given $x$ and $y$ one can iterate $|x-y|$
times, in each step picking up a factor $2dC_2/\lambda^s$ after a
maximum is taken over the $2d$ terms in the sums over next
neighbors. This results in a bound
\[
\E(|G_{\omega,\lambda}(x,y;z)|^s) \le \left( \frac{2dC_2}{\lambda^s}
\right)^{|x-y|} \sup_{u\in \Z^d} \E(|G_{\omega,\lambda}(x,u;z)|^s).
\]
For the last term we use the a-priori bound $C_1/\lambda^s$ provided
by Lemma~\ref{lem1}. We get the exponential decay in
(\ref{eq:fmdecay}) for $\lambda\ge \lambda_0$ if we choose
$\lambda_0$ such that $2dC_2/\lambda_0^s < 1$.

\end{proof}

We conclude this section by remarking that the exponential decay
bound found in Theorem~\ref{thm1} also holds for finite volume
restrictions of the Anderson Hamiltonian. More precisely, let $L\in
\N$ and $\Lambda_L := [-L,L]^d \cap \Z^d$. By
$h_{\omega,\lambda}^{\Lambda_L}$ and
$G_{\omega,\lambda}^{\Lambda_L}$ we denote the restriction of
$h_{\omega,\lambda}$ to $\ell^2(\Lambda_L)$ as well as its Green
function. By the same proof as above one finds that, for $\lambda
\ge \lambda_0$,
\begin{equation} \label{eq:finitevolgreen}
\E(|G_{\omega,\lambda}^{\Lambda_L}(x,y;z)|^s) \le C e^{-\mu|x-y|},
\end{equation}
where the constants $C<\infty$ and $\mu>0$ are now also uniform in
$L$.

Moreover, in the finite volume case the bound
(\ref{eq:finitevolgreen}) is uniform in $z\in \C$, allowing for real
energy. The reason for this is that the operators
$h_{\omega,\lambda}^{\Lambda_L}$ are finite-dimensional and that any
given real number $E$ is almost surely not one of their eigenvalues,
which implicitly follows from the above proof. In the finite volume
case this also holds for the a-priori bound in Lemma~\ref{lem1}.
This explains why such bounds play a role in the FMM similar to the
role played by Wegner estimates in localization proofs via MSA. They
demonstrate that eigenvalues are sensitive to the disorder
parameters. A good way to think of the main idea behind the FMM is
that eigenvalues are singularities of the resolvent which move
linearly under the random parameters. Thus the Green function can be
made integrable by taking fractional moments.

\section{From Fractional Moment Bounds to Localization}
\label{sec:Grafmethod}

We will now discuss methods which show that exponential decay of
fractional moments of Green's function as shown in
Theorem~\ref{thm1} implies spectral as well as dynamical
localization. For the sake of stating a general result of this form
we will absorb the disorder parameter into the random parameters
$\omega_x$ (re-scaling their distribution as in (\ref{eq:scaling})).
Thus we consider the Anderson Hamiltonian in its original form
(\ref{eq:Anderson}) with single-site distribution $\mu$ satisfying
(\ref{eq:density}).

From now on we will generally leave the dependence of various
quantities on the random variable $\omega$ implicit and write
$h=h_{\omega}$, $G=G_{\omega}$, etc.

Our goal is to prove

\begin{theorem} \label{thm2}
Let $I \subset \R$ be an open bounded interval. If there exist $s\in
(0,1)$, $C<\infty$ and $\mu>0$ such that
\begin{equation} \label{eq:fmdecay1}
\E(|G(x,y;E+i\varepsilon)|^s) \le C e^{-\mu|x-y|}
\end{equation}
uniformly in $E\in I$ and $\varepsilon>0$, then dynamical
localization in the form (\ref{eq:dynloc}) holds on the interval
$I$.
\end{theorem}

As a first consequence, by Theorem~\ref{thm1}  this implies that at
sufficiently large disorder $\lambda$ the Anderson model is
dynamically localized in the entire spectrum. In
Section~\ref{sec:bandedge} below, we will also use the criterion
provided by Theorem~\ref{thm2} to prove band edge localization.

The most direct way to conclude spectral localization, i.e.\ pure
point spectrum with exponentially decaying eigenfunctions, from
bounds such as (\ref{eq:fmdecay1}) is by the Simon-Wolff method. It
was developed in \cite{Simon/Wolff} to serve a similar purpose in
the context of multiscale analysis, where it showed that the Green
function bounds established in \cite{Frohlich/Spencer} indeed imply
spectral localization. A short argument, showing that the
Simon-Wolff criterion also can be combined with (\ref{eq:fmdecay1})
to show spectral localization, is provided in
\cite{Aizenman/Molchanov}.

Here we will instead discuss the proof of Theorem~\ref{thm2}, i.e.\
focus on how (\ref{eq:fmdecay1}) implies dynamical localization. We
have two reasons for doing so: First, dynamical localization is the
physically more relevant property. Second, as we will show at the
end of this section, dynamical localization implies spectral
localization with a straightforward argument using the RAGE theorem.

There are two substantially different arguments which prove
Theorem~\ref{thm2}. In this section we will present a modification
of an argument provided by Graf in \cite{Graf94}. This version of
the argument has recently also been used in \cite{HJS} to prove
dynamical localization for the so-called {\it unitary Anderson
model}.

The second method, via the use of eigenfunction correlators, will be
discussed in the next section.

Graf's argument starts with the realization that fractional moments
of Green's functions of the Anderson model can be used to bound the
second moment of Green's function as long as a small factor (the
imaginary part of the energy) is introduced to control the
singularities of Green's function at real energy.

\begin{proposition} \label{prop:secondmoment}
For every $s\in (0,1)$ there exists a constant $C_1<\infty$ only
depending on $s$ and $\rho$ such that
\begin{equation} \label{eq:secondmoment}
|\mbox{\rm Im}\,z| \,\E_x(|G(x,y;z)|^2) \le C_1 \E_x(|G(x,y;z)|^s)
\end{equation}
for all $z\in \C\setminus \R$ and $x, y\in \Z^d$.
\end{proposition}

Here $\E_x$ denotes averaging over $\omega_x$ as in the proof of
Lemma~\ref{lem1}. Integrating over the remaining variables, we see
that (\ref{eq:secondmoment}) also holds with $\E_x$ replaced by
$\E$. Our proof follows the proof of Lemma~3 in \cite{Graf94} almost
line by line.

\begin{proof}
As in the proof of Lemma~\ref{lem1} write $\omega = (\hat{\omega},
\omega_x)$. Keep $\hat{\omega}$ fixed and consider the Hamiltonian
\[ h^{(\alpha)} = h_{(\hat{\omega}, \omega_x+\alpha)} = h_{\omega} + \alpha P_{e_x}\]
obtained by ``wiggling the potential at $x$''. Its Green function
will be denoted by $G^{(\alpha)}$. Similar to (\ref{eq:resid}) to
(\ref{eq:solveG}) we find
\[
(h_{\omega}-z)^{-1} = (h^{(\alpha)}-z)^{-1} + \alpha
(h_{\omega}-z)^{-1} P_{e_x} (h^{(\alpha)}-z)^{-1},
\]
and
\begin{eqnarray} \label{eq:alphagreen}
G^{(\alpha)}(x,y;z) & = & \frac{G_{\omega}(x,y;z)}{1+\alpha G_{\omega}(x,x;z)} \\
& = & \frac{1}{\alpha + G_{\omega}(x,x;z)^{-1}} \cdot
\frac{G_{\omega}(x,y;z)}{G_{\omega}(x,x;z)}. \nonumber
\end{eqnarray}

For the special case $x=y$ and $\tilde{\alpha} = -
\mbox{Re}\,G_{\omega}(x,x;z)^{-1}$ we get from (\ref{eq:alphagreen})
that
\[ \left| \frac{1}{\mbox{Im}\,G(x,x;z)^{-1}} \right| = \left| G^{(\tilde{\alpha})}(x,x;z) \right| \le \frac{1}{|\mbox{Im}\,z|}, \]
i.e.\ $|\mbox{Im}\,G(x,x;z)^{-1}| \ge| \mbox{Im}\,z|$. Inserting
this into (\ref{eq:alphagreen}) gives
\begin{equation} \label{eq:firstbound}
|\mbox{Im}\,z| |G^{(\alpha)}(x,y;z)|^2 \le
\frac{|\mbox{Im}\,G_{\omega}(x,x;z)^{-1}|}{|\alpha +
G_{\omega}(x,x;z)|} \cdot
\frac{|G_{\omega}(x,y;z)|^2}{|G_{\omega}(x,x;z)|^2}.
\end{equation}

On the other hand, we can bound the same expression by
\begin{eqnarray} \label{eq:secondbound}
|\mbox{Im}\,z| |G^{(\alpha)}(x,y;z)|^2 & \le & |\mbox{Im}\,z| \sum_{y'\in \Z^d} |G^{(\alpha)}(x,y';z)|^2 \\
& = &| \mbox{Im}\,z| \langle e_x, (h^{(\alpha)}-\overline{z})^{-1} (h^{(\alpha)}-z)^{-1} e_x \rangle \nonumber \\
& = & |\mbox{Im}\,z| \langle e_x, \frac{1}{z-\overline{z}} [ (h^{(\alpha)}-z)^{-1} - (h^{(\alpha)}- \overline{z})^{-1} ] e_x \rangle \nonumber \\
& = & \left| \mbox{Im}\, G^{(\alpha)}(x,x;z) \right| \nonumber \\
& = & \frac{ |\mbox{Im}\, G_{\omega}(x,x;z)^{-1}|}{ |\alpha +
G_{\omega}(x,x;z)^{-1}|^2}, \nonumber
\end{eqnarray}
where the last step used (\ref{eq:alphagreen}) with $x=y$.

For $t\ge 0$ one has $\min(1,t^2) \le t^s$. Using this to
interpolate between (\ref{eq:firstbound}) and (\ref{eq:secondbound})
we get
\begin{equation} \label{eq:thirdbound}
|\mbox{Im}\,z| |G^{(\alpha)}(x,y;z)|^2 \le
\frac{|\mbox{Im}\,G_{\omega}(x,x;z)^{-1}|}{|\alpha +
G_{\omega}(x,x;z)^{-1}|^2} \cdot
\frac{|G_{\omega}(x,y;z)|^s}{|G_{\omega}(x,x;z)|^s}.
\end{equation}

We will now use the following ``re-sampling trick'', which has the
effect of creating an additional random variable (here $\alpha$) to
average over. For a non-negative Borel function $f$ on $\R$,
\begin{eqnarray} \label{eq:resampling}
\lefteqn{ \int \int f(\omega_x + \alpha) \rho(\omega_x+\alpha)\,d\alpha\, \rho(\omega_x)\,d\omega_x} \\
& = & \int \int f(\omega_x+\alpha) \rho(\omega_x+\alpha) \rho(\omega_x) \,d\omega_x \, d\alpha \nonumber \\
& = & \int \int f(\omega_x) \rho(\omega_x) \rho(\omega_x-\alpha) \,d\omega_x\, d\alpha \nonumber \\
& = & \int f(\omega_x) \rho(\omega_x) \left( \int \rho(\omega_x-\alpha)\,d\alpha \right)\,d\omega_x \nonumber \\
& = & \int f(\omega_x) \rho(\omega_x)\,d\omega_x, \nonumber
\end{eqnarray}
where the integration order was interchanged in the first and third
steps and translation invariance of Lebesgue measure was used in the
second.

Choose $f(\omega_x) = |G_{(\hat{\omega}, \omega_x)}(x,y;z)|^2$, then
(\ref{eq:resampling}) and (\ref{eq:thirdbound}) yield
\begin{eqnarray} \label{eq:fourthbound}
\lefteqn{ |\mbox{Im}\,z| \E_x (|G_{\omega}(x,y;z)|^2) } \\
& = & |\mbox{Im}\,z| \E_x \left( \int |G^{(\alpha)}(x,y;z)|^2 \rho(\omega_x+\alpha)\,d\alpha \right) \nonumber \\
& \le & \E_x \left( |\mbox{Im}\,G_{\omega}(x,x;z)^{-1}|
\frac{|G_{\omega}(x,y;z)|^s}{G_{\omega}(x,x;z)|^s} \int
\frac{\rho(\omega_x+\alpha)}{|\alpha + G_{\omega}(x,x;z)^{-1}|^2}
\,d\alpha \right). \nonumber
\end{eqnarray}

We now use Lemma~\ref{lem:graf} below with
$w=G_{\omega}(x,x;z)^{-1}$ to conclude
\[ |\mbox{Im}\, z| \E_x(|G_{\omega}(x,x;z)|^2) \le C \E_x (|G_{\omega}(x,y;z)|^s) \]
with a constant $C<\infty$ which only depends on supp$\,\rho$, but
not on $x$, $y$ and $z$.
\end{proof}

In the above proof we have used

\begin{lemma} \label{lem:graf}
There exists a constant $C=C(\rho)<\infty$ such that
\[ |\mbox{Im}\,w| \cdot |w|^s \int \frac{\rho(\omega_x+\alpha)}{|\alpha+w|^2} \,d\alpha \le C \]
uniformly in $w\in \C$ and $\omega_x \in \,\mbox{supp}\,\rho$.
\end{lemma}

\begin{proof}
Using $|w|^s \le |\alpha|^s + |\alpha+w|^s$, we need two estimates:

(i)
\begin{eqnarray*}
|\mbox{Im}\,w| \int \frac{|\alpha|^s \rho(\omega_x +\alpha)}{|\alpha+w|^2} \,d\alpha & \le &  \pi \| |\alpha|^s \rho(\omega_x+\alpha)\|_{\infty} \\
& \le & \pi (|\omega_x|^s \|\rho\|_{\infty} + \| |\lambda|^s
\rho(\lambda)\|_{\infty}).
\end{eqnarray*}

(ii)
\begin{eqnarray*}
|\mbox{Im}\,w| \int \frac{\rho(\omega_x+\alpha)}{|\alpha+w|^{2-s}} \,d\alpha & \le & \min( \frac{1}{|\mbox{Im}\,w|^{1-s}}, C \|\rho\|_{\infty} |\mbox{Im}\,w|^s) \\
& \le & C \|\rho\|_{\infty}^{1-s}.
\end{eqnarray*}

\end{proof}

We now complete the proof of Theorem~\ref{thm2}:

\begin{proof}
Consider the mixed spectral measures $\mu_{x,y}$ of $h$, the complex
Borel measures defined by
\begin{equation} \label{eq:specmeas}
\mu_{x,y}(B) = \langle e_x, \chi_B(h) e_y \rangle
\end{equation} for Borel sets $B\subset \R$. The total variation $|\mu_{x,y}|$ of $\mu_{x,y}$ is a regular bounded Borel measure which can be characterized by
\begin{eqnarray} \label{eq:totvar}
|\mu_{x,y}|(B) & = & \sup_{\small{\begin{array}{c} g:\R\to\C\;\mbox{Borel} \\ |g|\le 1 \end{array}}} \left| \int g(\lambda)\,d\mu_{x,y}(\lambda) \right| \\
& = & \sup_{|g|\le 1} |\langle e_x, g(h)\chi_B(h) e_y \rangle|,
\nonumber
\end{eqnarray}
e.g.\ \cite{Rudin}. The particular choice $g_t(x) = e^{-itx}$ in
(\ref{eq:totvar}) shows that
\begin{equation} \label{eq:dynlocspecial}
 |\mu_{x,y}|(I) \ge \sup_{t\in \R} |\langle e_x, e^{-ith} \chi_I(h) e_y \rangle |.
 \end{equation}
Therefore Theorem~\ref{thm2} will follow from a corresponding
exponential decay bound for $\E(|\mu_{x,y}|(I))$.

As $I$ is an open bounded interval, it follows from Lusin's Theorem
(\cite{Rudin}) that one can replace Borel functions in
(\ref{eq:totvar}) by continuous functions with compact support in
$I$,
\begin{equation} \label{eq:totvar2}
|\mu_{x,y}|(I) = \sup_{\small{\begin{array}{c} g \in C_c(I) \\
|g|\le 1 \end{array}}} |\langle e_x, g(h) e_y \rangle |.
\end{equation}

For $g\in C_c(I)$ it follows by elementary analysis (using that $g$ is bounded and uniformly continuous) that, uniformly in $\lambda \in \R$,
\[ g(\lambda) = \lim_{\varepsilon \to 0+} \frac{\varepsilon}{\pi} \int \frac{g(E)}{(\lambda -E)^2 +\varepsilon^2}\,dE.\]
By the spectral theorem this implies
\[ \langle e_x, g(h) e_y \rangle = \lim_{\varepsilon\to 0+} \frac{\varepsilon}{\pi} \int_I g(E) \langle e_x, (h-E-i\varepsilon)^{-1} (h-E+i\varepsilon)^{-1} e_y \rangle \,dE. \]
This allows to estimate the expected value of (\ref{eq:totvar2}) by
\begin{eqnarray*}
\lefteqn{\E(|\mu_{x,y}|(I))} \\
& \le & \E \left( \liminf_{\varepsilon\to 0+} \frac{\varepsilon}{\pi} \int_I \sum_{z\in \Z^d} | \langle e_x, (h-E-i\varepsilon)^{-1} e_z \rangle | |\langle e_z, (h-E+i\varepsilon)^{-1} e_y \rangle |\,dE \right) \\
& \le & \liminf_{\varepsilon\to 0+} \frac{1}{\pi} \int_I \sum_z \left( \E( \varepsilon |\langle e_x, (h-E-i\varepsilon)^{-1} e_z \rangle |^2 ) \right)^{1/2} \\
& & \mbox{} \cdot \left( \E( \varepsilon |\langle e_z,
(h-E+i\varepsilon)^{-1} e_y \rangle |^2) \right)^{1/2}\,dE,
\end{eqnarray*}
where, in this order, Fatou, Fubini and Cauchy-Schwarz (on $\E$)
have been used. Now Proposition~\ref{prop:secondmoment} can be
applied, allowing to bound further by
\begin{eqnarray*}
& \le & \liminf_{\varepsilon\to 0+} \frac{1}{\pi} \int_I \sum_z (\E(| G(x,z;E+i\varepsilon)|^s))^{1/2} (\E(| G(z,y;E-i\varepsilon)|^s )^{1/2} \,dE \\
& \le & \frac{C_1 C |I|}{\pi} \sum_z e^{-\mu |x-z|/2} e^{-\mu
|z-y|/2}.
\end{eqnarray*}
In the last step the assumption of Theorem~\ref{thm2} was used
(which also applies to $|G(z,y;E-i\varepsilon)| =
|G(y,z;E+i\varepsilon)|$). The elementary bound, based on the
triangle inequality,
\[ e^{-\mu|x-z|/2} e^{-\mu|z-y|/2} \le e^{-\mu|x-z|/4} e^{-\mu|x-y|/4} e^{-\mu|z-y|/4} \]
and another use of Cauchy-Schwarz (on the $z$-summation) complete
the proof of Theorem~\ref{thm2}.

\end{proof}

It deserves mentioning here that we have actually proven a stronger result than dynamical localization in the form (\ref{eq:dynloc}). The above proof shows that for an open interval $I$ on which (\ref{eq:fmdecay1}) holds there are constants $C<\infty$ and $\mu>0$ such that
\begin{equation} \label{eq:Borelloc}
\E(|\mu_{x,y}|(I)) = \E \big( \sup_{\tiny \begin{array}{cc} g:\R\to \C \,\mbox{Borel} \\ |g|\le 1 \end{array}} | \langle e_x, g(h) \chi_I(h) e_y \rangle | \big) \le C e^{-\mu|x-y|}
\end{equation}
for all $x,y \in \Z^d$. An interesting special case is $g=1$, where (\ref{eq:Borelloc}) establishes exponential decay of correlations in the spectral projection $\chi_I(h)$. Consequences of this for the conductivity of an electron gas in response to an electric field have been discussed in \cite{Aizenman/Graf}. Another consequence is mentioned at the end of this section.

Next we show that dynamical localization
implies pure point spectrum via the RAGE-Theorem. The underlying
idea is very simple: The RAGE-Theorem characterizes states in the
continuous spectral subspace as scattering states (in time-mean).
Dynamical localization excludes scattering states and thus
continuous spectrum.

\begin{proposition} \label{prop:dynspec}
Suppose that dynamical localization in the form (\ref{eq:dynloc})
holds in an open interval $I$. Then $h_{\omega}$ almost surely has
pure point spectrum in $I$.
\end{proposition}

\begin{proof}
For a discrete Schr\"odinger operator $h=h_0+V$ in $\ell^2(\Z^d)$
let $P_{cont}(h)$ be the projection onto its continuous spectral
subspace. Then the RAGE-Theorem, e.g.\ \cite{CFKS}, says that for
every $\psi \in \ell^2(\Z^d)$,
\begin{equation} \label{eq:rage}
\|P_{cont}(h) \chi_I(h) \psi\|^2 = \lim_{R\to\infty}
\lim_{T\to\infty} \int_0^T \frac{dt}{T} \|\chi_{\{|x|\ge R\}}
e^{-ith} \chi_I(h) \psi \|^2.
\end{equation}
If $\psi$ has finite support, say supp$\,\psi \subset \{|x|\le r\}$,
then
\begin{eqnarray*}
\|\chi_{\{|x|\ge R\}} e^{-ith} \chi_I(h) \psi\|^2 & \le &
\|\chi_{\{|x|\ge R\}} e^{-ith} \chi_I(h) \chi_{\{|x|\le r\}} \|
\|\psi\|^2 \\
& \le & \sum_{|x|\ge R, |y|\le r} | \langle e_x, e^{-ith} \chi_I(h)
e_y \rangle | \|\psi\|^2,
\end{eqnarray*}
where dropping a square is allowed as $\|\chi_{\{|x|\ge R\}}
e^{-ith} \chi_I(h) \chi_{\{|x|\le r\}} \| \le 1$.

Taking expectations in (\ref{eq:rage}) implies, after using Fatou
and Fubini,
\begin{eqnarray} \label{eq:ragebound}
\lefteqn{\E (\|P_{cont}(h_{\omega}) \chi_I(h_{\omega}) \psi\|^2)} \\
\nonumber & \le & \lim_{R\to\infty, T\to\infty} \int_0^T
\frac{dt}{T} \sum_{|x|\ge R, |y|\le r} \E\left(|\langle e_x,
e^{-ith_{\omega}} \chi_I(h_{\omega}) e_y \rangle |\right)
\|\psi\|^2.
\end{eqnarray}
By (\ref{eq:dynloc}) we have $\E(|\langle e_x, e^{-ith_{\omega}}
\chi_I(h_{\omega}) e_y \rangle|) \le Ce^{-\mu |x-y|}$ uniformly in
$t$, which bounds the right hand side of (\ref{eq:ragebound}) by
\[
\le \lim_{R\to\infty} \tilde{C} \sum_{|x|\ge R, |y|\le r}
e^{-\mu|x-y|} = 0.
\]
We conclude that $P_{cont}(h_{\omega})
\chi_I(h_{\omega}) \psi = 0$ for almost every $\omega$ and every
$\psi$ of finite support. The latter are dense in $\ell^2(\Z^d)$ and
thus $P_{cont}(h_{\omega}) \chi_{I}(h_{\omega}) =0$ almost surely,
meaning that the spectrum in $I$ is pure point.

\end{proof}

We note that the above proof of pure point spectrum does not imply exponential decay of corresponding
eigenfunctions. It is shown in \cite{Aizenman/Molchanov} how this follows directly from exponential decay of fractional moments (\ref{eq:fmdecay1}), using the Simon-Wolff-method \cite{Simon/Wolff}. It can also be deduced from (\ref{eq:Borelloc}) by considering $g(h)= \delta_E(h)$, $E\in I$, using the result from \cite{Simon:Cyclic} that almost surely all eigenvalues of  $h_{\omega}$ in $I$ are non-degenerate. For details on this see Section~2.5 of \cite{AENSS06}, where a corresponding argument for the continuum Anderson model is provided which also applies to the discrete Anderson model considered here.

\section{Finite Volume Methods} \label{sec:efcor}

The localization proof provided in Section~\ref{sec:Grafmethod} proceeds directly in {\it infinite volume}, i.e.\ does not require to consider restrictions of the hamiltonian $h$ to finite subsets of $\Z^d$. However, it is also possible to work in finite volume, prove the relevant bounds with volume-independent constants, and then deduce localization by taking the infinite volume limit (sometimes referred to as ``thermodynamical limit''). This has conceptual advantages such as having to only deal with discrete spectra, thus allowing to express functions of the hamiltonian by eigenfunction expansions and, as described at the end of Section~\ref{sec:largedisorder}, to directly study Green's function at real energy. Moreover, finite volume methods have proven very robust under generalizations, for example in the
extension to continuum Anderson models which we will discuss in
Section~\ref{sec:continuum}. For these reasons we will use this section to provide a different proof of dynamical localization, i.e.\ Theorem~\ref{thm2} above, using finite volume methods.

Many of the ideas involved here can already be found in the Kunz-Souillard approach to localization \cite{KS} for the one-dimensional Anderson model. They were first combined with the fractional moment method in \cite{Aizenman94} to prove dynamical localization for the multi-dimensional Anderson model. A central object are so-called {\it finite volume eigenfunction correlators}, arising from eigenfunction expansions. Eigenfunction correlators are
also used in similar form in proofs of dynamical localization via
multiscale analysis, see \cite{Stollmann} or \cite{Klein} and
references therein. 

The methods to be described here are not completely disjoint from the methods of Section~\ref{sec:Grafmethod}. As before, we consider the mixed spectral measures $\mu_{x,y}$ of $h$
introduced in (\ref{eq:specmeas}) as well as their total
variation $|\mu_{x,y}|$ given by (\ref{eq:totvar}). As will become clear in (\ref{eq:efcorbound}) below, $|\mu_{x,y}|$ can be considered as an {\it infinite volume eigenfunction correlator} for $h$. We will find bounds for it by finding bounds for finite volume eigenfunction correlators which hold uniformly in the volume.

Let 
$h_{\omega}^{\Lambda_L}$ be the restriction of $h_{\omega}$ to $\Lambda_L = [-L,L]^d
\cap \Z^d$ and denote its Green function by $G_{\omega}^{\Lambda_L}$.

\begin{proposition} \label{prop:totvarbound}
Let $0<s<1$ and $I$ an open bounded interval. Then there exists
$C=C(s,\rho,d)<\infty$ such that
\begin{equation} \label{eq:totvarbound}
\E(|\mu_{x,y}|(I)) \le C \liminf_{L\to\infty} \left( \int_I
\E(|G_{\omega}^{\Lambda_L}(x,y;E)|^s) \,dE \right)^{\frac{1}{2-s}}.
\end{equation}
\end{proposition}

Results of this form were first used in implicit form in \cite{Aizenman94} and later stated more explicitly in \cite{ASFH01}.
The exact statement given here as well as its proof
below and in Appendix~\ref{sec:appendix} follow notes provided to us in private communication by M.\ Aizenman and S.\ Warzel. They used similar results also in \cite{AW09}.

Based on (\ref{eq:dynlocspecial}), we see that
Proposition~\ref{prop:totvarbound} may be applied to provide a proof
of dynamical localization in $I$ in situations where it can be shown
that
\begin{equation} \label{eq:finitevolgreen2}
\E(|G_{\omega}^{\Lambda_L}(x,y;E)|^s) \le C e^{-\mu|x-y|}
\end{equation}
with constants which are uniform in $L$ and $E\in I$. In fact,
dynamical localization follows under the somewhat weaker assumption
that the energy average over $I$ of the fractional moments of
Green's function is exponentially decaying. However, in all our
applications we have uniform pointwise bounds available. For
example, as discussed at the end of Section~\ref{sec:largedisorder},
a bound of the form (\ref{eq:finitevolgreen2}) holds on the entire
spectrum for sufficiently large disorder, thus providing a second
proof of dynamical localization in this regime.

\begin{proof}[Proof of Proposition~\ref{prop:totvarbound}]
We start by reducing the claim (\ref{eq:totvarbound}) to properties
of finite-volume spectral measures. We again use the
characterization (\ref{eq:totvar2}) of $|\mu_{x,y}(I)|$ for open
bounded intervals $I$. Strong resolvent convergence of $h^L$ to $h$
implies for continuous $g$ of compact support that $\langle e_x,
g(h^{\Lambda_L}) e_y \rangle \to \langle e_x, g(h)e_y \rangle$ and
thus, by (\ref{eq:totvar2}),
\begin{equation} \label{eq:src}
|\mu_{x,y}|(I) \le \liminf_{L\to\infty} \sup_{|g|\le 1} |\langle
e_x, g(h^{\Lambda_L}) e_y \rangle |.
\end{equation}
Here the regularity assumption on $g$ can be dropped since
$h^{\Lambda_L}$ has discrete spectrum.

Let $h_x^L$ be the restriction of $h^{\Lambda_L}$ to the reducing subspace ${\mathcal H}_x$ for $h^{\Lambda_L}$ generated by $e_x$ and let $P_x$ be the orthogonal projection onto ${\mathcal H}_x$. Then $e_x$ is a cyclic vector for $h_x^L$ and all eigenvalues $E$ of $h_x^L$ are simple. Thus we may label the corresponding normalized eigenvectors by $\psi_E^L$. We use the notation $\psi_E^L$ also for $\psi_x^L \oplus 0$ in $\ell^2(\Lambda_L) = {\mathcal H}_x \oplus {\mathcal H}_x^{\perp}$.

By expanding into eigenvectors we get
\begin{eqnarray*}
|\langle e_x, g(h^{\Lambda_L}) e_y \rangle | & = & | \langle e_x, g(h_x^L) P_x e_y \rangle | \\
& = &
\Big| \sum_{E\in I\cap \sigma(h_x^L)} g(E)
\langle e_x, \psi_E^L \rangle \langle \psi_E^L, e_y \rangle \Big| \\
& \le & \sum_{E\in I\cap \sigma(h_x^L)} |\psi_E^L(x)||\psi_E^L(y)| \\
& =: & Q_L(x,y;I),
\end{eqnarray*}
and, in particular,
\[
\sup_{|g|\le 1} | \langle e_x, g(h^{\Lambda_L}) e_y \rangle | \le Q_L(x,y;I).
\]
The latter will be referred to as {\it eigenfunction
correlators}. Using Fatou's lemma we conclude from
(\ref{eq:src}) that
\begin{equation} \label{eq:efcorbound}
\E(|\mu_{x,y}|(I)) \le \liminf_{L\to\infty} \E(Q_L(x,y;I)).
\end{equation}

In order to establish a relation to the fractional moments of
Green's function we will also introduce {\it fractional
eigenfunction correlators} through
\begin{equation} \label{eq:fracefcor}
Q_L(x,y;I,r) := \sum_{E\in I\cap \sigma(h_x^L)}
|\psi_E^L(x)|^{2-r} |\psi_E^L(y)|^r
\end{equation}
for $0<r \le 2$, noting that $Q_L(x,y;I)=Q_L(x,y;I,1)$. We claim
that for $0<s<1$,
\begin{equation} \label{eq:interpol}
\E Q_L(x,y;I) \le (\E Q_L(x,y;I,s))^{\frac{1}{2-s}}.
\end{equation}
To see this, interpolate $s<1<2$ via $1=\frac{s}{p}+\frac{2}{q}$
with the conjugate exponents $p=2-s$ and $q=\frac{2-s}{1-s}$.
Applying H\"older to expectation as well as to summation yields
\[
\E Q_L(x,y;I,1) \le (\E Q_L(x,y;I,s))^{\frac{1}{2-s}} (\E Q_L(x,y;I,2))^{\frac{1-s}{2-s}}.
\]
This implies (\ref{eq:interpol}) after observing that $Q_L(x,y;I,2)
= \sum_{E\in I\cap \sigma(h_x^L)} |\psi_E^L(y)|^2 \le 1$.

We will now be able to relate the fractional eigenfunction
correlators to fractional moments of Green's function by showing
that there exists a constant $C=C(s,\rho,d)$ such that
\begin{equation} \label{eq:greenefcor}
\E Q_L(x,y;I,s) \le C \int_I \E(|G^{\Lambda_L}(x,y;E)|^s)\,dE.
\end{equation}
This, combined with (\ref{eq:efcorbound}) and (\ref{eq:interpol}),
implies (\ref{eq:totvarbound}).

In the proof of (\ref{eq:greenefcor}) we will use the fractional
eigenfunction correlators $Q_{L,v}(x,y;I,s)$ which are defined as in
(\ref{eq:fracefcor}), but with the summation being over the
eigenvalues and eigenfunctions of $h_x^L+vP_{e_x}$. Note that, as $e_x$ is a cyclic vector for $h_x^L$, $h_x^L+v P_{e_x}$ is the same as the restriction of $h^{\Lambda_L} +vP_{e_x}$ to ${\mathcal H}_x$ and that $e_x$ is a cyclic vector for this operator for all values of $v\in \R$. This makes Proposition~\ref{prop:efcor} in Appendix~\ref{sec:appendix} applicable to our situation, which we will now use to finish the proof of Proposition~\ref{prop:totvarbound} by invoking a {\it resampling
argument}.

For this note that $\int
\frac{\rho(u)}{|u-\alpha|^s}\,du$ is continuous and non-vanishing as
a function of $\alpha\in \R$. Thus there exists a constant
$C=C(s,\rho)<\infty$ such that
\begin{equation} \label{eq:resample}
\frac{\rho(\alpha)}{\int \frac{\rho(u)}{|u-\alpha|^s}\,du} \le C
\quad \mbox{for all $\alpha\in \R$}.
\end{equation}

Writing $\omega = (\hat{\omega}, \omega_x)$ and denoting the
expectation over $\hat{\omega}$ by $\hat{\E}$, we get
\begin{eqnarray} \label{eq:resample2}
\E(Q_L^{\omega}(x,y;I,s)) & = & \hat{\E} \int_{\R}
Q_L^{(\hat{\omega}, \omega_x)}(x,y;I,s) \rho(\omega_x)\,d\omega_x \\
\nonumber & \le & C\hat{\E} \int_{\R} \left( \int_R
Q_L^{(\hat{\omega}, \omega_x)}(x,y;I,s)
\frac{d\omega_x}{|u-\omega_x|^s} \right) \rho(u)\,du.
\end{eqnarray}
After the change of variable $\omega_x \mapsto v:=\omega_x-u$ we see
that the inner integral is equal to $\int
Q_{L,v}^{(\hat{\omega},u)}(x,y;I,s)\,\frac{dv}{|v|^s}$. By
Proposition~\ref{prop:efcor} this coincides with $\int_I
|G_{(\hat{\omega},u)}^{\Lambda_L}(x,y;E)|^s\,dE$. Inserting into
(\ref{eq:resample2}) we arrive at
\begin{eqnarray*}
\E(Q_L^{\omega}(x,y;I,s)) & \le & C \hat{\E} \int_{\R} \int_I |G_{(\hat{\omega}, u)}^{\Lambda_L}(x,y;E)|^s\,dE \\
& = & C \E \int_I |G_{\omega}^{\Lambda_L}(x,y;E)|^s\,dE.
\end{eqnarray*}

\end{proof}

\section{Lifshits Tails and Band Edge Localization} \label{sec:bandedge}

\subsection{Band edge localization}

So far the only regime in which we have proven localization is the case of large disorder $\lambda >>1$ in (\ref{eq:Andersondis}). In this section we consider the Anderson model in its original form (\ref{eq:Anderson}), i.e.\ at fixed disorder. Our goal is a proof of localization at energies near the bottom of the spectrum. The arguments involved can be easily modified to show the same near the upper spectral edge.

For notational convenience we will assume that the density $\rho$ of the distribution of the $\omega_x$ satisfies supp$\,\rho = [0,\omega_{max}]$. We also write $E_0=-2d$, which according to Theorem~\ref{thm:charspectrum} becomes the bottom of the almost sure spectrum
\[ \Sigma = [E_0, 2d+\omega_{max}] \]
of $h_{\omega}$.

Our localization proof will again proceed via showing exponential decay of the fractional moments of Green's function:

\begin{theorem} \label{thm3}
For every $s\in (0,1)$ there exist $\delta>0$, $\mu>0$ and
$C<\infty$ such that
\begin{equation} \label{eq:fmdecay2}
\E(|G_{\omega}(x,y;E+i\epsilon)|^s) \le Ce^{-\mu|x-y|}
\end{equation}
for all $x,y \in \Z^d$, $E\in [E_0, E_0+\delta)$ and $\epsilon>0$.
\end{theorem}

As discussed at the end of Section~\ref{sec:largedisorder}, our
methods again yield a bound on the finite volume Green function as
in (\ref{eq:finitevolgreen}), uniform in the volume and in $E\in
[E_0,E_0+\delta]$, allowing for $\epsilon=0$. Thus we can conclude spectral and dynamical localization at the bottom of the spectrum from either of the methods discussed in Sections~\ref{sec:Grafmethod} of \ref{sec:efcor}, working in infinite volume or in finite volume.

Compared to the case of large disorder, we face essentially two new
difficulties, which are illustrated by the shortcomings of
Lemmas~\ref{lem1} and \ref{lem2}. The a-priori bound from
Lemma~\ref{lem1} is still valid and will be used. But, as the
disorder $\lambda$ is fixed, we can not hope that the a-priori bound
also provides a ``smallness mechanism'', which can be used to
iteratively prove exponential decay. We will again proceed by
iteration, but a different mechanism will be needed to get it
started. Also, it will turn out that we need a different decoupling
argument. Lemma~\ref{lem2}, which was used in this context in the
proof of Theorem~\ref{thm1}, is too case-specific and will not work
for the geometric situations which we will encounter here.

\subsection{Lifshits tails}

Physically, the new smallness mechanism is provided by the fact that
the bottom of the spectrum $E_0 =\inf \Sigma$ is a {\it fluctuation
boundary}. This means that finite volume restrictions of
$h_{\omega}$ rarely have eigenvalues close to $0$. To make this precise, as
before let $h_{\omega}^{\Lambda_L}$ be the restriction of
$h_{\omega}$ to $\ell^2(\Lambda_L)$, $\Lambda_L = [-L,L]^d \cap
\Z^d$.

\begin{lemma} \label{lem4}
For every $\beta\in (0,1)$ there are $\eta>0$ and $C<\infty$ such
that
\begin{equation} \label{eq:fluctbound}
\PP(\inf \sigma(h_{\omega}^{\Lambda_L}) \le E_0 + L^{-\beta}) \le CL^d
e^{-\eta L^{\beta d/2}}
\end{equation}
for all $L\in \N$.
\end{lemma}

To illustrate why this means that small eigenvalues are rare, let us
assume that the $|\Lambda_L| =(2L+1)^d$ eigenvalues of
$h_{\omega}^L$ are uniformly distributed in $\Sigma$. Then the
smallest eigenvalue should be no larger than  $C/L^d$. But by
(\ref{eq:fluctbound}) this is extremely rare for large $L$. In fact,
the methods used to prove (\ref{eq:fluctbound}) can also be used to
prove that the integrated density of states $N(E)$ of $h_{\omega}$
satisfies Lifshits-tail asymptotics near the bottom of the spectrum:
\begin{equation} \label{eq:lifshits}
N(E) \le C e^{-\eta |E-E_0|^{-d/2}},
\end{equation}
which is much ``thinner'' near $E_0$ than the corresponding IDS
$N_0(E) = C|E-E_0|^{d/2}$ of the Laplacian $h_0$.

For detailed proofs of Lemma~\ref{lem4} as well as
(\ref{eq:lifshits}) we refer to \cite{Kirsch} or \cite{Stollmann},
with the latter working in the setting of the continuum Anderson
model (but applicable to the discrete model as well). Here we only
briefly outline the reasons behind Lemma~\ref{lem4}. By the
variational principle
\begin{eqnarray} \label{eq:variation}
\inf \sigma(h_{\omega}^{\Lambda_L}-E_0) & = & \inf_{\|\varphi\|=1}
\langle (h_{\omega}^{\Lambda_L}-E_0) \varphi, \varphi \rangle \\ \nonumber
& = & \inf_{\|\varphi\|=1} \Big( \langle (h_0^{\Lambda_L}-E_0) \varphi,
\varphi \rangle + \sum_{i\in \Lambda_L} \omega_i |\varphi(i)|^2
\Big).
\end{eqnarray}

Note that both terms on the right hand side of (\ref{eq:variation}),
the kinetic and potential energy, are non-negative. In order to find
a low lying eigenvalue, they both need to be small. By reasons of
the uncertainty principle, small kinetic energy requires that
$\varphi$ is approximately constant, $\varphi(i) \sim C =
|\Lambda_L|^{-1/2}$, to be normalized. For such $\varphi$ the
potential energy is approximately $\sum_{i\in \Lambda_L} \omega_i/
|\Lambda_L|$, which by the central limit theorem with large probability is close to the expected
value $\E(\omega_0)>0$. The event $\sum_i \omega_i /|\Lambda_L| <
L^{-\beta} <\E(\omega_0)$ is a large deviation and has probability
exponentially small in $|\Lambda_L|$.

The weakest part of the above heuristics is the reference to the
uncertainty principle. Slightly different ways to make this rigorous
are provided in \cite{Kirsch} and \cite{Stollmann}, both requiring
arguments which control the separation of the second lowest
eigenvalue from the lowest eigenvalue. The proof provided in
\cite{Kirsch} (going back to work from the 1980s) uses Temple's
inequality in this context, while \cite{Stollmann} uses an argument
based on analytic perturbation theory.

In the context of proving Theorem~\ref{thm3}, Lemma~\ref{lem4}
provides a first step, a so-called {\it initial length estimate}:

\begin{lemma}[Initial Length Estimate] \label{lem5}
For every $s\in (0,1)$ there exist $C<\infty$ and $\eta>0$ such that
\begin{equation} \label{eq:ile}
\E(|G^{\Lambda_L}(x,y;E+i\epsilon)|^s) \le CL^d e^{-\eta
L^{d/(d+2)}}
\end{equation}
for all $L\in \N$, $x$, $y \in \Lambda_L$ with $|x-y| \ge L/2$, $E\in
[E_0,E_0+\frac{1}{2}L^{-2/(d+2)}]$ and $\epsilon>0$.
\end{lemma}

\begin{proof}
Let $\beta \in (0,1)$ and, motivated by Lemma~\ref{lem4}, define the
``good'' and ``bad'' sets as $\Omega_G := \{\omega: \inf
\sigma(h_{\omega}^{\Lambda_L}-E_0) \le L^{-\beta}\}$ and $\Omega_B :=
\Omega_G^c$. Then
\begin{equation} \label{eq:goodbad}
\E(|G^{\Lambda_L}(x,y;E+i\epsilon)|^s) =
\E(|G^{\Lambda_L}(x,y;E+i\epsilon)|^s \chi_{\Omega_G}) +
\E(|G^{\Lambda_L}(x,y;E+i\epsilon)|^s \chi_{\Omega_B}).
\end{equation}

Pick $p>1$ sufficiently small such that $sp<1$ and let $q$ be
conjugate to $p$, $\frac{1}{p}+\frac{1}{q}=1$. H\"older applied to
the second term on the right hand side of (\ref{eq:goodbad}) gives
\begin{eqnarray} \label{eq:bad}
\E(|G^{\Lambda_L}(x,y;E+i\epsilon)|^s \chi_{\Omega_B}) & \le &
\left(
\E(|G^{\Lambda_L}(x,y;E+i\epsilon)|^{sp} )\right)^{1/p}  \PP(\Omega_B)^{1/q} \\
\nonumber & \le & CL^{d/q} e^{-\frac{\eta}{q} L^{\beta d/2}},
\end{eqnarray}
where we have used the a-priori bound from Lemma~\ref{lem1} as well
as the probability bound from Lemma~\ref{lem4}. The first term on
the right hand side of (\ref{eq:goodbad}) concerns the event where
$E$ has distance at least $\frac{1}{2}L^{-\beta}$ from the bottom of
the spectrum, which allows to conclude exponential decay of
$|G^{\Lambda_L}(x,y;E+i\epsilon)|$ in $|x-y|$ from a Combes-Thomas
estimate (see e.g.\ \cite{Kirsch}), giving
\begin{eqnarray} \label{eq:good}
\E(|G^{\Lambda_L}(x,y;E+i\epsilon)|^s \chi_{\Omega_G}) & \le & C
L^{\beta s} e^{-s\eta|x-y|/L^{\beta}} \\ \nonumber & \le & C
L^{\beta s} e^{-\eta_1 L^{1-\beta}}
\end{eqnarray}
for constants $\eta>0$, $\eta_1>0$ and $C<\infty$.

The choice $\beta = 2/(2+d)$ leads to equal exponents in
(\ref{eq:bad}) and (\ref{eq:good}), which combine to give
(\ref{eq:ile}).
\end{proof}

\subsection{Geometric decoupling}

We will eventually fix $L=L_0$, choosing $L_0$ such that
the right hand side of (\ref{eq:ile}) is sufficiently small (how
small still to be determined). After making this choice we will pick
$\delta = \frac{1}{2}L_0^{-2/(d+2)}$, thus determining the interval
$[E_0,E_0+\delta]$ in which Theorem~\ref{thm3} establishes localization.
In order to derive the exponential decay bound (\ref{eq:fmdecay2})
from this we have to develop a decoupling method which will allow to
proceed iteratively, splitting the path from $x$ to $y$ into
segments of length $L_0$.

The description of this so-called {\it geometric decoupling method}
will fill the remainder of this section. Our argument will closely
follow a construction introduced in \cite{ASFH01}.

In addition to $h_{\omega}^{\Lambda_L}$, consider
$h_{\omega}^{\Lambda_L^c}$, the restriction of $h_{\omega}$ to
$\ell^2(\Lambda_L^c)$, where $\Lambda_L^c = \Z^d \setminus
\Lambda_L$. Let
\[ h_{\omega}^{(L)} = h_{\omega}^{\Lambda_L} \oplus h_{\omega}^{\Lambda_L^c}.\]
This means that
\begin{equation} \label{eq:decouple}
h_{\omega} = h_{\omega}^{(L)} + T^{(L)},
\end{equation}
where $T^{(L)}$ is the operator containing the ``hopping terms''
introduced by the discrete Laplacian between sites of $\Lambda_L$
and $\Lambda_L^c$. We write $G_{\omega}^{(L)}(z) = (h_{\omega}^{(L)}-z)^{-1}$.

More precisely, the matrix-elements of $T^{(L)}$ are
\begin{equation} \label{eq:hopping}
T^{(L)}(u,u') = \left\{ \begin{array}{ll} -1, & \mbox{if $(u,u') \in
\Gamma_L$}, \\ 0, & \mbox{else}, \end{array} \right.
\end{equation}
where $\Gamma_L$ is the boundary of $\Lambda_L$ defined as the set
of pairs $(u,u')$ with $|u-u'|=1$ and $u\in \Lambda_L$, $u'\in
\Lambda_L^c$, or vice versa.

We now perform a {\it double decoupling} of the resolvent
$G_{\omega} = G_{\omega}(z)$ by using the
resolvent equation for (\ref{eq:decouple}) twice, first at $L$ and
then at $L+1$:

\begin{eqnarray} \label{eq:doubledecoup}
G_{\omega} & = & G_{\omega}^{(L)} - G_{\omega}^{(L)} T^{(L)}
G_{\omega} \\ \nonumber & = & G_{\omega}^{(L)} - G_{\omega}^{(L)}
T^{(L)} G_{\omega}^{(L+1)} + G_{\omega}^{(L)} T^{(L)} G_{\omega}
T^{(L+1)} G_{\omega}^{(L+1)}.
\end{eqnarray}

Equations of this form are often referred to as {\it geometric resolvent
identities}.

By translation invariance it suffices to prove (\ref{eq:fmdecay2})
for $x=0$. If $|y|\ge L+2$, then the first two terms on the right
hand side of (\ref{eq:doubledecoup}) do not contribute to the
matrix-element $G_{\omega}(0,y;z)$ and thus
\begin{eqnarray*}
G_{\omega}(0,y;z) & = & \langle e_0, G_{\omega}^{(L)} T^{(L)}
G_{\omega} T^{(L+1)} G_{\omega}^{(L+1)} e_y \rangle \\ \nonumber & =
& \sum_{(u,u')\in \Gamma_L} \sum_{(v,v')\in \Gamma_{L+1}}
G_{\omega}^{(L)}(0,u;z) G_{\omega}(u',v;z)
G_{\omega}^{(L+1)}(v',y;z).
\end{eqnarray*}

For $s\in (0,1)$ we get
\begin{eqnarray} \label{eq:expdecouple}
\lefteqn{\E(|G_{\omega}(0,y;z)|^s)} \\ \nonumber & \le &
\sum_{\begin{array}{cc} (u,u')\in \Gamma_L \\ (v,v')\in \Gamma_{L+1}
\end{array}} \E \left( |G_{\omega}^{\Lambda_L}(0,u;z)
G_{\omega}(u',v;z) G_{\omega}^{\Lambda_{L+1}^c}(v',y;z)|^s \right).
\end{eqnarray}
Here we have replaced $G_{\omega}^{(L)}$ by $G_{\omega}^{\Lambda_L}$
as $0$ and $u$ are both in $\Lambda_L$. Similarly,
$G_{\omega}^{(L+1)}$ was replaced by $G_{\omega}^{\Lambda_{L+1}^c}$.
For fixed $(u,u')$ and $(v,v')$ consider the corresponding term on
the right hand side of (\ref{eq:expdecouple}) and note that the first
and last of the three factors are independent of $\omega_{u'}$ and
$\omega_v$. Thus, in taking the expectation we may integrate over
$\omega_{u'}$ and $\omega_v$ first and use Lemma~\ref{lem1} to
conclude
\begin{eqnarray} \label{eq:bigtrick}
\lefteqn{\E \left( |G_{\omega}^{\Lambda_L}(0,u;z) G_{\omega}(u',v;z) G_{\omega}^{\Lambda_{L+1}^c}(v',y;z)|^s \right)} \\
& \le & C \E \left( |G_{\omega}^{\Lambda_L}(0,u;z)|^s |G_{\omega}^{\Lambda_{L+1}^c}(v',y;z)|^s \right) \label{eq:bt2} \\
& = & C \E ( |G_{\omega}^{\Lambda_L}(0,u;z)|^s ) \E (
|G_{\omega}^{\Lambda_{L+1}^c}(v',y;z)|^s ). \label{eq:bt3}
\end{eqnarray}
In the last step we have used that the remaining two factors in
(\ref{eq:bt2}) are stochastically independent. Now let
$z=E+i\epsilon$ with $E\in [E_0, E_0+\frac{1}{2}L^{-2/(2+d)}]$. Then we may
estimate the first factor in (\ref{eq:bt3}) by the bound obtained in
Lemma~\ref{lem5} and, after inserting into (\ref{eq:expdecouple}),
find
\begin{equation} \label{eq:firststep}
\E(|G_{\omega}(0,y;z)|^s) \le CL^{2d-1} e^{-\eta L^{d/(d+2)}}
\sum_{\|v'\|_{\infty}=L+2} \E (
|G_{\omega}^{\Lambda_{L+1}^c}(v',y;z)|^s ).
\end{equation}

We want to use (\ref{eq:firststep}) as the first step in an iteration. The second step would consist in finding a bound for $\E (|G_{\omega}^{\Lambda_{L+1}^c}(v',y;z)|^s$ similar to the bound for $\E(|G_{\omega}(0,y;z)|^s)$ given by (\ref{eq:firststep}), with $v'$ serving as the new origin. A problem arises from the fact that the underlying domain is not any longer $\Z^d$, but $\Lambda_{L+1}^c$. Iterating this would result in more and more complex geometries and we would be faced with the problem to check if all the constants involved in the estimates leading to (\ref{eq:firststep}) can be chosen uniform in those geometries.

An elegant way around this is the following result of \cite{ASFH01}, see Lemma~2.3 there, which allows to bound the depleted Green function $G_{\omega}^{(L+1)}$ in terms of the full Green function $G_{\omega}$:

\begin{lemma} \label{lem:fullbound}
There exists a constant $C=C(s,\rho)<\infty$ such that
\[
\E(|G_{\omega}^{(L+1)}(v',y;z)|^s) \le \E(|G_{\omega}(v',y;z)|^s) + C \sum_{\|u'\|_{\infty} = L+2} \E(|G_{\omega}(u',y;z)|^s).
\]
\end{lemma}

The proof of this starts from the geometric resolvent identity $G_{\omega}^{(L+1)} = G_{\omega} + G_{\omega}^{(L+1)} T^{(L+1)} G_{\omega}$. In the resulting Green function expansion over $(u,u') \in \Gamma_{L+1}$ crucial use is made of the bound
\[
\E(|G_{\omega}^{(L+1)}(v',u;z)|^s |G_{\omega}(u',y;z)|^s) \le C \E(|G_{\omega}(u',y;z)|^s).
\]
The proof of this uses another special case of the Krein formula similar to (\ref{eq:krein}) (but this time tracking the dependence on all four variables $\omega_{v'}$, $\omega_u$, $\omega_{u'}$ and $\omega_y$) and a decoupling argument similar to Lemma~\ref{lem2} (but for functions of two variables which are linear separately in each variable). For the details we refer to \cite{ASFH01}.

Lemma~\ref{lem:fullbound} combines with (\ref{eq:firststep}) to yield
\begin{equation} \label{eq:itstart}
\E(|G_{\omega}(0,y;z)|^s) \le C L^{4d} e^{-\eta L^{d/(d+2)}} \sup_{\|u_1\|_{\infty} \le L+2} \E(|G_{\omega}(u_1,y;z)|^s)
\end{equation}
for some constant $C<\infty$. With that constant we fix $L=L_0$ such that $\rho:= C L_0^{4d} e^{-\eta L_0^{d/(d+2)}} <1$. We also
choose $\delta := \frac{1}{2} L_0^{-2/(2+d)}$ now. For $E\in
[E_0, E_0+\delta]$ we can use (\ref{eq:itstart}) to start an iteration,
\[ \E (|G_{\omega}(u_1,y;z)|^s) \le \rho \sup_{\|u_2\|_{\infty} \le 2(L_0+2)} \E(|G_{\omega}(u_2,y;z)|^s), \]
and so forth. This iteration can be carried out approximately $|y|/L_0$ times before the chains $u_1$, $u_2$, \ldots may reach $y$. After this number of steps we use the a-priori bound
from Lemma~\ref{lem1} to bound the last fractional moment in the
chain. We have proven Theorem~\ref{thm3} with exponential decay rate
$\mu = |\log \rho|/L_0$.

\section{The Continuum Anderson Model} \label{sec:continuum}

It took somewhat more than a decade to find a generalization of the
fractional moment method to continuum Anderson models. Our goal in
this section is to explain why this took so long and how it was
eventually done. Here our presentation will be less self-contained than in previous sections. We will outline the new ideas which were needed and refer to the literature for details.

The main difficulty is that the rather elementary arguments from
rank-one and rank-two perturbation theory, which worked so well for
the discrete Anderson model, fall far short of applying in the
continuum. In the latter, each single site potential is a
perturbation of infinite rank, which at best has certain compactness properties relative to the Laplacian. To make the
central ideas behind the fractional moment method work in this
setting required a much deeper understanding of some of the
operator-theoretic aspects involved. Here we will follow the works
\cite{AENSS06} and \cite{BNSS06}, where these questions were
settled. Earlier work in \cite{HKK} extended certain aspects of the
fractional moment method to continuum models, but still relied on
finite-rank perturbation arguments by, for example, considering
continuum models with random point interactions.

For our presentation here we choose to work with the deterministic
background operator
\begin{equation} \label{eq:contbackground}
H_0 = -\Delta + V_0
\end{equation}
in $L^2(\R^d)$, where $V_0$ is a real-valued, $\Z^d$-periodic
potential in $L^{\infty}(\R^d)$. Let $E_0 := \inf \sigma(H_0)$
denote its spectral minimum.

A continuum Anderson-type model is then given by
\begin{equation} \label{eq:contAnderson}
H_{\omega} = H_0 - \sum_{n\in \Z^d} \omega_n U_n,
\end{equation}
where $\omega = (\omega_n)_{n\in \Z^d}$ is an array of i.i.d.\
random variables with bounded density $\rho$ such that supp$\,\rho =
[0,\omega_{max}]$.

The single-site potentials $U_n(x) = U(x-n)$ are translates of a
non-negative bump function $U$ characterized by the existence of
$0<r_1 \le r_2<\infty$ and $0< c_1 \le c_2 <\infty$ such that
\begin{equation} \label{eq:bump}
c_1 \chi_{\{|x|\le r_1\}} \le U \le c_2 \chi_{\{|x|\le r_2\}}.
\end{equation}

The spectrum of $H_{\omega}$ is almost surely deterministic,
\[ \sigma(H_{\omega}) = \Sigma \quad \mbox{a.s.}, \]
and
\[ E_1 := \inf \Sigma = \inf \sigma(H_0-\omega_{max}\sum_n U_n)\]
is characterized by choosing all couplings maximal and thus, due to
our sign-convention, the potential minimal. It can be shown under
the assumption (\ref{eq:bump}) that the spectral minimum is strictly
decreased by the random potential: $E_1<E_0$.

We will use the notation $\chi_n =\chi_{\Lambda_1(n)}$, where
$\Lambda_1(n)$ refers to the unit cube in $\R^d$ centered at $n\in
\Z^d$.

The following theorem is a special case of a result in
\cite{BNSS06}. Similar results were first obtained in
\cite{AENSS06}, where a ``covering condition'' of the form
\begin{equation} \label{eq:cover}
U\ge c\chi_0, \quad c>0,
\end{equation}
was required for the single-site potential.

\begin{theorem} \label{thm4}
Let $d\le 3$ and $0<s<\frac{1}{3}$. Then there exist $\delta>0$,
$\mu>0$ and $C<\infty$ such that
\begin{equation} \label{eq:contloc}
\E (\|\chi_k (H_{\omega}-E-i\epsilon)^{-1} \chi_{\ell}\|^s) \le C
e^{-\mu |k-\ell|}
\end{equation}
for all $E\in [E_1,E_1+\delta]$, $\epsilon>0$ and $k, \ell \in
\Z^d$.
\end{theorem}

In Theorem~\ref{thm4} we use the norm of the localized resolvent
$\chi_k (H_{\omega}-E-i\epsilon)^{-1} \chi_{\ell}$ (sometimes called
a ``smeared Green function'') as a continuum analogue of the
discrete Green function $G(x,y;E+i\epsilon)$. This has also been
found to be the correct object to consider in continuum extensions
of multiscale analysis.

Without going into the details here (which for the continuum case
can be done similar to what was described in
Section~\ref{sec:efcor}, see \cite{AENSS06}), we state that
exponential decay of fractional moments of the smeared Green
function, as established in (\ref{eq:contloc}), implies spectral and
dynamical localization:

\begin{coro} \label{cor:contloc}
Under the assumptions of Theorem~\ref{thm4} the following holds:

(a) For almost every $\omega$, $H_{\omega}$ has pure point spectrum
in $[E_1, E_1+\delta]$ with exponentially decaying eigenfunctions.

(b) There are constants $\mu>0$ and $C<\infty$ such that
\begin{equation} \label{eq:contdynloc}
\E \left( \sup_{|g|\le 1} \|\chi_k g(H_{\omega})
\chi_{[E_1,E_1+\delta]}(H_{\omega}) \chi_{\ell} \| \right) \le
Ce^{-\mu |k-\ell|}
\end{equation}
for all $k, \ell \in \Z^d$, with the supremum taken over Borel
functions $g:\R\to\C$.
\end{coro}

The overall approach to proving Theorem~\ref{thm4} is similar to the
proof of Theorem~\ref{thm3} in the previous section. The main steps
are:

\vspace{.3cm}

(i) {\bf A priori-bound:} It can be shown that to every $E_2 \in
(E_1,E_0)$ and $0<s<1$ there exists $C<\infty$ such that
\begin{equation} \label{eq:aprioricont}
\E (\|\chi_k (H_{\omega}-E-i\epsilon)^{-1} \chi_{\ell}\|^s) \le C
\end{equation}
uniformly in $E\in [E_1,E_2]$, $\epsilon>0$ and $k, \ell \in \Z^d$.

Note here that, as opposed to the discrete case Lemma~\ref{lem1},
the a-priori bound is only shown for energies below the spectrum of
the unperturbed operator $H_0$. This is a consequence of not
requiring the covering condition (\ref{eq:cover}) for the single-site potential. If a covering
condition holds, then it was shown in \cite{AENSS06} that the bound
(\ref{eq:aprioricont}) holds at all energies, with a constant $C$ on
the right which grows polynomially in $E$.

\vspace{.3cm}

(ii) {\bf Lifshits tails:} The bottom $E_1$ of the almost sure spectrum is again a
fluctuation boundary and close analogues to Lemmas~\ref{lem4} and
\ref{lem5} as well as Lifshits tail asymptotics (\ref{eq:lifshits}) of the IDS hold in
the continuum, see e.g.\ \cite{Stollmann} and \cite{AENSS06}. As in
the discrete case, this provides the start of an iterative procedure
for the proof of exponential decay in (\ref{eq:contloc}).

\vspace{.3cm}

(iii) {\bf Geometric decoupling:} The geometric decoupling procedure described at the end of
Section~\ref{sec:bandedge} can be carried out similarly in the
continuum. Additional technical difficulties arise mostly due to the
fact that the required geometric resolvent identities (compare
(\ref{eq:doubledecoup})) are less straightforward in the continuum.
One consequence of this is the restriction of Theorem~\ref{thm4} to
$s<1/3$, which is due to the need of an additional three-factor
H\"older bound used in the decoupling procedure. Also, elementary decoupling bounds such as Lemma~\ref{lem2} have to be replaced by a more systematic construction involving {\it resampling} of the random variables $\omega_n$ near the surfaces at which the decoupling is carried out. For details in the
setting of Theorem~\ref{thm4} see \cite{BNSS06}.

\vspace{.3cm}

The only one of the above three points which we want to address in
some more detail is the a-priori bound (\ref{eq:aprioricont}), as
the existence of such a bound can be seen as the crucial test for
the possibility of using the fractional moment method in the
continuum.

For simplicity, we only consider the ``diagonal'' case $k=\ell =0$
here and will assume the covering condition (\ref{eq:cover}). We
will discuss reasons why we could hope that
\begin{equation} \label{eq:Ubound}
\sup_{\varepsilon>0} \E(\|U(H_{\omega}-E-i\epsilon)^{-1}U\|^s) < \infty
\end{equation}
for energies near $\inf \Sigma$. Under the covering condition, this implies the same result with
$U$ replaced by $\chi_0$.

When trying to implement ideas similar to the ones used in the proof
of Lemma~\ref{lem1}, we are faced with having to find an analogue to
the Krein formula. It turns out that this is done by the identities
known from Birman-Schwinger theory. Write
\[ \omega = (\hat{\omega}, \omega_0), \quad H_{\omega} =
H_{\hat{\omega}} -\omega_0 U.\]

Then, at least formally, it is easy to derive by the resolvent
identity that
\begin{equation} \label{eq:BS}
U^{1/2} (H_{\omega}-z)^{-1} U^{1/2} = (A_{BS}-\omega_0 I)^{-1}
\end{equation}
in $L^2(\mbox{supp}\,U)$, with the Birman-Schwinger operator
\begin{equation} \label{eq:BSO}
A_{BS} = \left( U^{1/2} (H_{\hat{\omega}}-z)^{-1} U^{1/2}
\right)^{-1}.
\end{equation}

It can be justified that the inverses in (\ref{eq:BSO}) and
(\ref{eq:BS}) exist and that $A_{BS}$ is maximally dissipative. Here
an operator $A$ is called maximally dissipative if it is
dissipative, i.e.\ Im$\langle \phi, A\phi \rangle \ge 0$ for all
$\phi$ in its domain, and it has no proper dissipative extension.
This can also be characterized by the fact that $\{e^{itA}\}_{t\ge
0}$ is a contraction semigroup.

The identity (\ref{eq:BS}) looks promising since the right hand side
separates the dependence on $\omega_0$ from the dependence on
$\hat{\omega}$. Indeed, if the bound (\ref{eq:2by2fracbound}) could
be generalized from dissipative $2\times 2$-matrices to general
maximally dissipative operators $B$, then it would immediately give
us (\ref{eq:Ubound}). While (\ref{eq:2by2fracbound}) extends to
dissipative $N\times N$-matrices, the bound $C(r,s)$ on the right
will become $N$-dependent and diverge for $N\to\infty$, as is seen
by choosing $B$ to be a diagonal matrix with entries $1, \ldots, N$.
Thus it is not possible to directly extend (\ref{eq:2by2fracbound})
to the Hilbert space setting.

However, the extension to the Hilbert space setting becomes possible
if additional Hilbert-Schmidt multipliers are introduced. This is
most naturally stated in terms of a closely related
weak-$L^1$-bound:

\begin{theorem} \label{thm5}
Let ${\mathcal H}_0$ and ${\mathcal H}_1$ be separable Hilbert
spaces, let $A$ be maximally dissipative in ${\mathcal H}_0$, and
let $M:{\mathcal H}_0 \to {\mathcal H}_1$ be a Hilbert-Schmidt
operator. Then

(a) the boundary value
\[
M(A-v+i0)^{-1} M^* := \lim_{\epsilon\to 0} M(A-v+i\epsilon)^{-1} M^*
\]
exists in Hilbert-Schmidt norm for almost every $v\in \R$,

(b) there exists a constant $C < \infty$ (independent of $A$ and
$M$) such that
\begin{equation} \label{eq:weakL1}
|\{ v\in \R: \,\|M(A-v+i0)^{-1} M^*\|_{HS} > t\}| \le \frac{C
\|M\|_{HS}^2}{t}
\end{equation}
for all $t>0$.
\end{theorem}

In (\ref{eq:weakL1}) $|\cdot|$ denotes Lebesgue measure and
$\|\cdot\|_{HS}$ the Hilbert-Schmidt norm.

Part (a) is well known in mathematical physics and has been
frequently used in scattering theory. Less well known is part (b),
which describes the value-distribution of the boundary values
guaranteed to exist by (a). This was proven in \cite{Naboko} (see
also an appendix in \cite{AENSS06} for a reproduction of the proof),
based on the weak-$L^1$-property of the Hilbert transform of Hilbert
space-valued functions, the latter being a quite classical result in
harmonic analysis.

The weak-$L^1$-bound (\ref{eq:weakL1}) can be turned into the
fractional moment bound
\begin{equation} \label{eq:HSfmbound}
\int \|M(A-v+i0)^{-1} M^*\|_{HS}^s \rho(v)\,dv \le C(s,\rho)
\|M\|_{HS}^{2s},
\end{equation}
where the constant $C(s,\rho)$ can be chosen uniform for all
Hilbert-Schmidt operators $M$ and maximally dissipative $A$. This is
done by the standard layer-cake integration argument: If $F(v) :=
\|M(A-v+i0)^{-1}M^*\|_{HS}$, then
\begin{eqnarray*}
\int |F(v)|^s \rho(v)\,dv & \le & \|\rho\|_{\infty}
\int_{\mbox{supp}\,\rho} |F(v)|^s\,dv \\
& = & \|\rho\|_{\infty} \int_0^{\infty} |\{v\in \mbox{supp}\,\rho:
|F(v)|^s >t\}|\,dt.
\end{eqnarray*}
By (\ref{eq:weakL1}) the integrand is bounded by $\min \{\tilde{C},
C\|M\|_{HS}^2/t^{1/s}\}$, where $\tilde{C} = |\mbox{supp}\,\rho|$.
Splitting the integral at the $t$-value where $\tilde{C} =
C\|M\|_{HS}^2/t^{1/s}$ leads to (\ref{eq:HSfmbound}).

When trying to use (\ref{eq:HSfmbound}) for a proof of
(\ref{eq:Ubound}) we see from (\ref{eq:BS}) that
\begin{equation} \label{eq:applfmbound}
U(H_{\omega}-z)^{-1} U = U^{1/2} (A_{BS}-\omega_0 I)^{-1} U^{1/2}.
\end{equation}

This leaves us with one more problem to deal with: The
multiplication operator $U^{1/2}$ in $L^2(\R^d)$ is {\it not}
Hilbert-Schmidt. In fact, multiplication operators with
non-vanishing functions in the continuum are {\it never} compact.

The key to solving this last problem is that $U^{1/2}$ is {\it
relatively} Hilbert-Schmidt with respect to $-\Delta$ (meaning that
$U^{1/2}(-\Delta+1)^{-1}$ is Hilbert-Schmidt), at least for $d\le
3$, see e.g.\ \cite{Simon}. Arguments as typical in relative
perturbation theory allow to split the left hand side of
(\ref{eq:applfmbound}) into a sum of terms, some of which trivially
satisfy a fractional moment bound, while others include additional
multipliers which lead to the Hilbert-Schmidt property required in
(\ref{eq:HSfmbound}). These arguments only work at energies below
the spectrum of the unperturbed operator $H_0$, which is the reason
for the corresponding assumption which we made when stating
(\ref{eq:aprioricont}).

For further details on these relative perturbation arguments as well
as on the ``off-diagonal'' case $k\not= \ell$ in
(\ref{eq:aprioricont}) we refer to \cite{BNSS06} and conclude our
sketch of the proof of Theorem~\ref{thm4} here.

\section{Open Problems and directions for future work} \label{sec:problems}

To conclude this introduction into the theory of Anderson
localization, we mention some open problems and discuss some wide
open issues which mathematicians need to understand better in the
future. Here we will not restrict ourselves to further developments
of the fractional moment method, but will address broader aspects of
the quantum mechanical description of disordered media. We will be
relatively brief here and note that a more complete and more
detailed recent discussion of open problems in this field can be found in
\cite{Banff}. In particular, we do not attempt here to give complete
references to related works.

\subsection{Singular distributions} \label{sec:singular}

Consider the discrete and continuous Anderson models $h_{\omega}$
and $H_{\omega}$, but allow for singular distributions of the random
coupling parameters $\omega_i$, $i\in \Z^d$. The most extreme case
would be the case of independent Bernoulli variables, i.e.\
$\PP(\omega_i=a)=p$, $\PP(\omega_i=b)=1-p$. This models the
physically interesting case of a two-component alloy. Both, the
fractional moments method and the Fr\"ohlich-Spencer multiscale
analysis, fail to provide localization proofs in this situation. The
reason for this is that both methods to a large extend use {\it
local averaging} arguments in the random parameters, as demonstrated
very clearly by the proof of Lemma~\ref{lem1} above. While it is
possible to deal with H\"older-continuous distributions, the
Bernoulli case it out of reach for the traditional approaches.

However, Bourgain and Kenig \cite{Bourgain/Kenig} have shown

\begin{theorem} \label{thm6}
Consider the continuum Anderson model $H_{\omega}$ defined by
(\ref{eq:contbackground}) and (\ref{eq:contAnderson}) with $V_0=0$
and independent Bernoulli random variables $(\omega_i)$. Then
$H_{\omega}$ is spectrally localized near $E_1 =\inf \Sigma$.
\end{theorem}

Their proof is based on a substantial enhancement of the multiscale
analysis approach and, in particular, a deeper understanding of the
underlying averaging mechanisms (such as the role of the so-called
{\it Wegner estimates}). It has also been shown in \cite{AGKW09} how
the argument provided in \cite{Bourgain/Kenig} can be used to handle
$(\omega_i)$ with arbitrary non-trivial distribution.

However, somewhat surprisingly, the same question remains open for
the {\it discrete} Anderson model (\ref{eq:Anderson}) with Bernoulli
distributed random couplings. The technical reason for this is that
\cite{Bourgain/Kenig} uses subtle unique continuation properties of
the eigenfunctions of Schr\"odinger operators which are not
available for lattice models.

More generally, one can easily imagine various other models of
random operators where the random parameters naturally have discrete
distribution and where the available mathematical methods fail to
prove localization. One such model would be discrete Laplacians on
random subgraphs of the edges of $\Z^d$. An open question is to
decide if in the supercritical percolation regime, where the graph
has a unique infinite component, the Laplacian has localized
spectrum. For a recent survey on these models see
\cite{Muller/Stollmann}

\subsection{Extended states}

Every list of open problems in random operator theory needs to
mention the somewhat embarrassing fact that mathematicians are still
far from understanding the physically conjectured extended states
regime in the three-dimensional Anderson model.

A proof of the existence of continuous (or absolutely continuous)
spectrum or of diffusive solutions to the time-dependent
Schr\"odinger equation for this model would be an important
break-through. Here we would like to mention another way to
characterize the Anderson transition from localized to extended
states, namely the {\it level statistics conjecture}. In fact, this
is how physicists can numerically distinguish the two regimes, which
provides the most important evidence for the correctness of the
physical heuristics explaining the transition.

According to the level statistics conjecture it is possible to
distinguish the localized and delocalized regimes by considering the
statistical distribution of the eigenvalues (viewed as point
processes) of finite volume restrictions of the Anderson model.
Localized states should be characterized by Poisson statistics of
the eigenvalues, while in spectral regions with extended states the
finite volume eigenvalues should show GOE statistics. The latter it
a special kind of level repulsion observed for Gaussian orthogonal
ensembles in random matrix theory.

In the spectral regimes where mathematicians can establish
localization, it has also been verified that the finite volume
eigenvalues are Poisson distributed, see \cite{Molchanov},
\cite{Minami}, \cite{CGK} and \cite{Germinet/Klopp}. However, regarding GOE statistics in
the Anderson model, as little is known as for other possible
characterizations of extended states.

As discussed in the lectures by L.\ Erdos at this School
\cite{Erdos}, GOE statistics is a rather universal phenomenon
observed in large classes of random matrices, e.g.\ so-called {\it
Wigner random matrices}. The most apparent difference between Wigner
matrices and the Anderson model is that for the latter randomness is
restricted to the diagonal matrix-elements while in Wigner matrices
all entries are random. Understanding the transition between
Anderson models and random matrices, for example by considering {\it
random band matrices} with an increasing amount of off-diagonal
random entries, could provide important insights into the
localization-delocalization transition in the Anderson model.

\subsection{Electron-electron interactions and many-body systems}

The Anderson models discussed above are {\it one-electron} models,
which ignore interactions between electrons (as well as interactions
between nuclei, which are considered as affixed to the lattice
sites). Quite recently, Anderson-type models for a fixed number $N$
of interacting electrons in a random background have been shown to
have localization properties. Chulaevsky and Suhov \cite{CS1, CS2}
have done this by an extension of multiscale analysis, while
\cite{AW09} povides similar results based on the fractional moments
approach.

Let us give one example of a result which can be obtained by both
approaches, where we do not try to state the most general result. An
$N$-particle discrete Anderson-type model can be defined as
\[
(h_{\omega}^{(N)} \phi)(x) = \sum_{y:|y-x|=1} \phi(y) + (U(x) +
\lambda\sum_{j=1}^N \omega_{x_j}) \phi(x),
\]
where $\phi \in \ell^2(\Z^{Nd})$, $x=(x_1,\ldots,x_N) \in \Z^{Nd}$
and $y\in \Z^{Nd}$. As above, the $(\omega_x)_{x\in\Z^d}$ are
i.i.d.\ random variables with bounded, compactly supported density.
Assume, for simplicity, that $U(x)$ is a two-particle interaction
term of finite range,
\[
U(x) = \sum_{1\le j<k\le N} \Phi(x_j-x_k), \quad \mbox{supp$\,\Phi$
finite}.
\]

Then localization holds for large disorder:

\begin{theorem} \label{thm7}
If $\lambda$ is sufficiently large, then $h_{\omega}^{(N)}$ is
spectrally and dynamically localized at all energies.
\end{theorem}

Note that a version of dynamical localization suitable for
$N$-particle systems has to be used here, see \cite{AW09}.

While such results provide important mathematical progress,
condensed matter physicists will object to the above model because
it keeps the number of electrons fixed. The physically correct
system to look at would be a model with positive {\it electron
density}. For example, one could consider the restriction of
$h_{\omega}^{(N)}$ to $\ell^2(\Lambda_L^N)$, where $\Lambda_L^N =
\Lambda_L \times \ldots \times \Lambda_L$, and study its properties
if $N$ goes to infinity together with the volume, i.e. $N \sim L^d
\to \infty$. This is mathematically wide open. It is not even clear
how the concepts of localization and extended states should be
defined in this setting, with spectral theoretic terms most likely
not being the correct language any more. Instead dynamical
descriptions will have to be used. One particularly interesting
question, which is still challenging even to physicists, is if
electron interaction enhances or reduces localization effects. It
has been argued in the physics literature that such effects could be
crucial to understand localization for two-dimensional disordered
systems, the critical case.

\begin{appendix}

\section{Basic Rank-One Perturbation Theory} \label{sec:appendix}

Here we first provide some facts from rank-one perturbation theory. They can be considered as finite-dimensional, elementary versions of some of the much more profound insights behind Simon-Wolff theory, e.g.\ \cite{Simon/Wolff} or Chapter 11 to 13 of \cite{Simon:Trace}. Then we establish a general relation between fractional moments of Green's function and eigenfunction correlators, which was used in Section~\ref{sec:efcor} above.

Let $\mathcal H$ be a finite-dimensional Hilbert space,
dim$\,{\mathcal H} = N$, $h_0$ a self-adjoint operator in $\mathcal
H$, and $\varphi$ a normalized cyclic vector for $h_0$. By $P =
\langle \varphi, \cdot \rangle \varphi$ we denote the orthogonal
projection onto span$\{\varphi\}$. Consider the family of rank-one
perturbations
\[
h_v := h_0 +vP, \quad v\in \R,
\]
of $h_0$. Then $\varphi$ is a cyclic vector for all $h_v$. All $h_v$
have simple eigenvalues, which we label as
\[
E_1(v) < E_2(v) < \ldots < E_N(v).
\]
Also consider the self-adjoint operator
\[
h_{\infty} := Ph_0 P \quad \mbox{on $D(h_{\infty}) =
\{\varphi\}^{\perp}$}.
\]
Application of the Gram-Schmidt procedure to $\varphi, A\varphi, A^2
\varphi, \ldots$ shows that $A\varphi - \langle \varphi, A\varphi
\rangle \varphi$ is a cyclic vector for $h_{\infty}$. In particular,
$h_{\infty}$ has simple eigenvalues which we label as
\[
E_1(\infty) < E_2(\infty) < \ldots < E_{N-1}(\infty).
\]
In the following we will also use the notation $E_0(\infty):=
-\infty$, $E_N(\infty) := \infty$.

\begin{lemma} \label{lem:A1}
(i) For every $k\in \{1,\ldots,N\}$, $E_k(v)$ is analytic and
strictly increasing as a function of $v\in \R$.

(ii) For every $v\in \R$, the eigenvalues of $h_v$ and $h_{\infty}$
are intertwined as
\begin{equation} \label{eq:intertwine}
E_1(v) < E_1(\infty) < E_2(v) < E_2(\infty) < \ldots < E_{N-1}(v) <
E_{N-1}(\infty) < E_N(v).
\end{equation}

(iii) For every $k\in \{1,\ldots,N\}$,
\[
\lim_{v\to -\infty} E_k(v) = E_{k-1}(\infty), \quad
\lim_{v\to\infty} E_k(v) = E_k(\infty).
\]
\end{lemma}

\begin{proof}
(i) Let $\psi_k(v)$ denote normalized eigenvectors of $h_v$ to
$E_k(v)$, $k=1,\ldots,N$. That $\varphi$ is a cyclic vector for
$h_v$ means that $\langle \psi_k(v), \varphi \rangle \not= 0$ for
all $k\in \{1,\ldots, N\}$. By analytic perturbation theory, the
functions $E_k(v)$ are analytic with
\begin{equation} \label{eq:1OPT}
 E_k'(v) = \langle \psi_k(v), P \psi_k(v) \rangle = |\langle
\psi_k(v), \varphi \rangle |^2 \not= 0.
\end{equation}

(ii) We will prove this by the variational characterization of
eigenvalues of $h_v$ and $h_{\infty}$, e.g.\ Theorem XIII.2 in
\cite{Reed/SimonIV}, which says that for $k\in \{1,\ldots,N\}$,
\begin{equation} \label{eq:minmax}
E_k(v) = \sup_{\tiny \begin{array}{c} V\subset {\mathcal H} \\
\mbox{dim}\,V=k-1 \end{array}} \inf_{\tiny \begin{array}{c} f\in
V^{\perp} \\ \|f\|=1 \end{array}} \langle f, h_v f \rangle
\end{equation}
and for $k\in \{1,\ldots, N-1\}$,
\begin{equation} \label{eq:minmax2}
E_k(\infty) = \sup_{\tiny \begin{array}{c} \tilde{V}\in
\{\varphi\}^{\perp} \\ \mbox{dim}\,\tilde{V}=k-1 \end{array}}
\inf_{\tiny \begin{array}{c} g\in \tilde{V}^{\perp} \\ \|g\|=1
\end{array}} \langle g, h_{\infty} g \rangle.
\end{equation}
Note that in the infimum in (\ref{eq:minmax}) the orthogonal complement $V^{\perp}$ is taken with respect to $\mathcal H$, while $\tilde{V}^{\perp}$ in (\ref{eq:minmax2}) is taken with respect to $\{\varphi\}^{\perp}$. By definition of $h_{\infty}$ and $h_v$ we also have that
\begin{equation} \label{eq:idperp}
\langle f, h_{\infty} f \rangle = \langle f, h_v f \rangle \quad \mbox{for all $f\in \{\varphi\}^{\perp}$}.
\end{equation}

We first show that
\[
E_k(v) \le E_k(\infty) \quad \mbox{for all $k\in \{1,\ldots,N\}$}.
\]
This is trivial for $k=N$. For $k\le N-1$, let $V$ be a subspace of $\mathcal H$ with dim$\,V=k-1$. Then $V^{\perp} \cap \{\varphi\}^{\perp}$ is a subspace of $\{\varphi\}^{\perp}$ of dimension at least $N-k$. Thus it is the orthogonal complement of a subspace $W$ of $\{\varphi\}^{\perp}$ of dimension at most $k=1$. Therefore, by (\ref{eq:minmax2}) and (\ref{eq:idperp}),
\begin{eqnarray*}
E_k(\infty) & \ge & \inf_{\tiny \begin{array}{c} f\in W^{\perp} \\ \|f\|=1 \end{array}} \langle f, h_{\infty} f \rangle = \inf_{\tiny \begin{array}{c} f\in V^{\perp} \cap \{\varphi\}^{\perp} \\ \|f\|=1 \end{array}} \langle f, h_v f \rangle \\
& \ge & \inf_{\tiny \begin{array}{c} f\in V^{\perp} \\ \|f\|=1 \end{array}} \langle f, h_v f \rangle.
\end{eqnarray*}
As this holds for every subspace $V$ of $\mathcal H$ with dim$\,V=k-1$, (\ref{eq:minmax}) implies $E_k(\infty) \ge E_k(v)$.

Next we will show that
\[
E_k(\infty) \le E_{k+1}(v) \quad \mbox{for all $k\in \{0, \ldots, N-1\}$},
\]
which is trivial for $k=0$. Let $k\ge 1$ and $\tilde{V} \subset \{\varphi\}^{\perp}$ with dim$\,\tilde{V} = k-1$. Then $V:= \,\mbox{span}\{\varphi\} \oplus \tilde{V} \subset {\mathcal H}$ with dim$\,V=k$ and $V^{\perp} = \{0\} \oplus \tilde{V}^{\perp}$. Thus, by (\ref{eq:minmax}) and (\ref{eq:idperp}),
\[
E_{k+1}(v) \ge \inf_{\tiny \begin{array}{c} f\in V^{\perp} \\ \|f\|=1 \end{array}} \langle f, h_v f \rangle = \inf_{\tiny \begin{array}{c} g\in \tilde{V}^{\perp} \\ \|g\|=1 \end{array}} \langle g, h_{\infty} g \rangle.
\]
As $\tilde{V} \subset \{\varphi\}^{\perp}$ with dim$\,\tilde{V} = k-1$ was arbitrary, (\ref{eq:minmax2}) implies $E_{k+1}(v) \ge E_k(\infty)$.

Strictness of all inequalities in (\ref{eq:intertwine}) now is a consequence of (i).

(iii) Here we use the following general fact, which can be proven using Schur complementation (see e.g.\ \cite{BHS} for a description of this method): For the self-adjoint $2\times 2$-block matrix
\[ \left( \begin{array}{cc} A & B \\ B^* & D \end{array} \right),\]
let $E\not\in \sigma(D)$, then
\[
\lim_{|v|\to \infty} \left( \begin{array}{cc} A+vI-EI & B \\ B^* & D-EI \end{array} \right)^{-1} = \left( \begin{array}{cc} 0 & 0 \\ 0 & (D-EI)^{-1} \end{array} \right).
\]
Applying this to the $2\times 2$-block representation of $h_0$ in span$\{\varphi\} \oplus \{\varphi\}^{\perp}$ shows that $(h_v-EI)^{-1} \to 0 \oplus (h_{\infty}-EI)^{-1}$ as $|v|\to \infty$ for every $E\not\in \sigma(h_{\infty})$.

Using that for self-adjoint operators $A$,
\[ \|(A-EI)^{-1}\| = \frac{1}{\mbox{dist}(E,\sigma(A))},\]
we conclude that for every $E\in \sigma(h_{\infty})$ there exists a function $E(v)$ such that $E(v) \in \sigma(h_v)$ for all $v$ and $\lim_{v\to\infty} E(v) = E$. If $E=E_k(\infty)$ for $k=1,\ldots, N-1$, it follows from the results of (i) and (ii) that $E(v) = E_k(v)$ for $v$ sufficiently large, i.e.\ $\lim_{v\to\infty} E_k(v)=E_k(\infty)$. Similarly, it follows that $\lim_{v\to -\infty} E_{k+1}(v) = E_k(\infty)$. $E_1(v) \to -\infty$ as $v\to -\infty$ and $E_N(v) \to\infty$ as $v\to\infty$ follows easily by minimizing/maximizing the quadratic form of $h_v$.

\end{proof}

To given $\chi \in {\mathcal H}$, $s\in (0,1)$ and open interval $I\subset \R$ we consider the fractional eigenfunction correlators
\begin{equation} \label{eq:defefcor}
Q_v(\varphi,\chi;I,s) := \sum_{k:E_k(v)\in I} |\langle \psi_k(v), \varphi \rangle|^{2-s} |\langle \psi_k(v), \chi \rangle|^s.
\end{equation}

\begin{proposition} \label{prop:efcor}
The fractional eigenfunction correlators satisfy the identity
\begin{equation} \label{eq:efcorfm}
\int_{\R} \frac{Q_v(\varphi, \chi;I,s)}{|v|^s}\,dv = \int_I |\langle \varphi, (h_0-E)^{-1} \chi \rangle|^s\,dE.
\end{equation}

\end{proposition}

\begin{proof}
As $\varphi$ is cyclic for $h_v$ and thus $\langle \psi_k(v), \varphi \rangle \not= 0$ for all $k$, we can rewrite (\ref{eq:defefcor}) as
\begin{equation} \label{eq:efcor2}
Q_v(\varphi,\chi;I,s) = \sum_{k:E_k(v) \in I} E_k'(v) \left| \frac{ \langle \psi_k(v), \chi \rangle}{ \langle \psi_k(v), \varphi \rangle} \right|^s,
\end{equation}
where we have also used (\ref{eq:1OPT}). Observe that
\[ (h_0-E_k(v)) \psi_k(v) = (h_v -E_k(v) -vP) \psi_k(v) = -v \langle \varphi, \psi_k(v) \rangle \varphi.\]
Thus, using that $E_k(v) \not\in \sigma(h_0)$ for all $v\not= 0$,
\begin{eqnarray*}
\langle \psi_k(v), \chi \rangle & = & \langle (h_0-E_k(v)) \psi_k(v), (h_0-E_k(v))^{-1} \chi \rangle \\
& = & -v \langle \psi_k(v), \varphi \rangle \langle \varphi, (h_0-E_k(v))^{-1} \chi \rangle.
\end{eqnarray*}
This allows to further rewrite (\ref{eq:efcor2}) as
\begin{equation} \label{eq:efcor3}
\frac{Q_v(\varphi, \chi;I,s)}{|v|^s} = \sum_{k:E_k(v) \in I} E_k'(v) | \langle \varphi, (h_0-E_k(v))^{-1} \chi \rangle |^s.
\end{equation}
This will allow to prove (\ref{eq:efcorfm}) by integration. Here we may assume that $I\subset (E_{k-1}(\infty), E_k(\infty))$ for a fixed $k\in \{1,\ldots, N\}$ (from which the general case follows easily). In this case (\ref{eq:efcor3}) says that
\[
\frac{Q_v(\varphi, \chi;I,s)}{|v|^s} = \left\{ \begin{array}{ll} E_k'(v) |\langle \varphi, (h_0-E_k(v))^{-1} \chi \rangle |^s, & \mbox{if $E_k(v)\in I$}, \\ 0, & \mbox{else.} \end{array} \right.
\]
Integration yields
\begin{eqnarray*}
\int_{\R} \frac{Q_v(\varphi, \chi;I,s)}{|v|^s} \,dv & = & \int_{v:E_k(v)\in I} E_k'(v) |\langle \varphi, (h_0-E_k(v))^{-1} \chi \rangle |^s \,dv \\
& = & \int_I | \langle \varphi, (h_0-E)^{-1} \chi \rangle|^s \,dE,
\end{eqnarray*}
which used the substitution $v \mapsto E=E_k(v)$.
\end{proof}

We conclude by noting that it is possible to prove a result
corresponding to Proposition~\ref{prop:efcor} without assuming that
the Hilbert space is finite-dimensional, as long as the spectral
measure for $h_0$ corresponding to $\chi$ is purely singular and
boundary values $\langle \varphi, (h_0-E-i0)^{-1} \chi \rangle$ of
Green's function are used. A corresponding streamlining of
Aizenman's original arguments in \cite{Aizenman94} has been provided
by Simon in \cite{SimonAT} for the unitary models considered there
and will be presented for the self-adjoint setting in
\cite{AizenmanWarzelBook}.

\end{appendix}

\end{document}